\documentclass[english,journal]{IEEEtran}
\usepackage[T1]{fontenc}
\usepackage[latin9]{inputenc}
\usepackage{color}
\usepackage{float}
\usepackage{amsmath}
\usepackage{amsthm}
\usepackage{amssymb}
\usepackage{graphicx}

\makeatletter

\floatstyle{ruled}
\newfloat{algorithm}{tbp}{loa}
\providecommand{\algorithmname}{Algorithm}
\floatname{algorithm}{\protect\algorithmname}

  \theoremstyle{plain}
  \newtheorem{lem}{\protect\lemmaname}
  \theoremstyle{plain}
  \newtheorem{prop}{\protect\propositionname}
\theoremstyle{plain}
\newtheorem{thm}{\protect\theoremname}

\usepackage{amsfonts}
\usepackage{cite
}\usepackage{array}
\usepackage{algorithm}
\usepackage{algorithmic}
\usepackage{subfigure}
\usepackage{stfloats}

\newcounter{MYtempeqncnt}

\makeatother

\usepackage{babel}

\makeatother

\usepackage{babel}
  \providecommand{\lemmaname}{Lemma}
  \providecommand{\propositionname}{Proposition}
\providecommand{\theoremname}{Theorem}

\begin{document}

\title{Distributed Optimization of Hierarchical Small Cell Networks: A GNEP
Framework}

\author{Jiaheng~Wang,~Wei~Guan,~Yongming~Huang, Robert~Schober, Xiaohu~You\thanks{J. Wang, W. Guan, Y. Huang, and X. You are with the National Mobile
Communications Research Laboratory, Southeast University, Nanjing,
China (email: \{jhwang, weiguan, huangym, xhyu\}@seu.edu.cn).

R. Schober is with the Institute for Digital Communications, Friedrich-Alexander-University
Erlangen-Nürnberg (FAU), Erlangen, Germany (email: robert.schober@fau.de). }}
\maketitle
\begin{abstract}
Deployment of small cell base stations (SBSs) overlaying the coverage
area of a macrocell BS (MBS) results in a two-tier hierarchical small
cell network. Cross-tier and inter-tier interference not only jeopardize
primary macrocell communication but also limit the spectral efficiency
of small cell communication. This paper focuses on distributed interference
management for downlink small cell networks. We address the optimization
of transmit strategies from both the game theoretical and the network
utility maximization (NUM) perspectives and show that they can be
unified in a generalized Nash equilibrium problem (GNEP) framework.
Specifically, the small cell network design is first formulated as
a GNEP, where the SBSs and MBS compete for the spectral resources
by maximizing their own rates while satisfying global quality of service
(QoS) constraints. We analyze the GNEP via variational inequality
theory and propose distributed algorithms, which only require the
broadcasting of some pricing information, to achieve a generalized
Nash equilibrium (GNE). Then, we also consider a nonconvex NUM problem
that aims to maximize the sum rate of all BSs subject to global QoS
constraints. We establish the connection between the NUM problem and
a penalized GNEP and show that its stationary solution can be obtained
via a fixed point iteration of the GNE. We propose GNEP-based distributed
algorithms that achieve a stationary solution of the NUM problem at
the expense of additional signaling overhead and complexity. The convergence
of the proposed algorithms is proved and guaranteed for properly chosen
algorithm parameters. The proposed GNEP framework can scale from a
QoS constrained game to a NUM design for small cell networks by trading
off signaling overhead and complexity.
\end{abstract}


\section{Introduction}

The proliferation of smart phones and mobile Internet applications
has caused an explosive growth of wireless services that will saturate
the current cellular networks in the near future \cite{Ericsson13}.
Small cells, including femtocells, picocells, and microcells, have
been widely viewed as a key enabling technology for next generation
(5G) mobile networks \cite{Andrews14}. By densely deploying low-power
low-cost small cell base stations (SBSs) in addition to traditional
macrocell base stations (MBSs), small cells can offload data traffic
from macrocells, improve local coverage, and achieve higher spectral
efficiency \cite{Andrews12}.

The coexistence of SBSs and MBSs results in a two-tier hierarchical
heterogeneous network architecture \cite{Hossain14}. To fully exploit
the potential of small cells, full frequency reuse among small cells
and macrocells is needed \cite{Quek13}, which presents several difficulties
in network design. First, there exist both cross-tier interference
between small cells and macrocells and inter-cell interference among
small cells (or macrocells). Second, the two tiers have different
service requirements. As fundamental infrastructure, MBSs provide
basic signal coverage for both existing and emerging services so that
macrocell communication generally has high priority and strict quality
of service (QoS) requirements. SBSs are mainly deployed as supplements
of MBSs to offload data traffic from MBSs and provide wireless access
for local users at as high rates as possible \cite{Andrews12,Hossain14}.

In small cell networks, QoS satisfaction of macrocell users (MUEs)
is jeopardized by cross-tier interference from SBSs, especially when
the network is operated in a closed access manner, where each cell
only serves registered users. In this case, an MUE near a small cell
may experience strong interference from the SBS \cite{Andrews12}.
Meanwhile, in the absence of regulation, small cell users (SUEs) also
suffer from inter-tier interference from other SBSs, which are often
densely deployed, as well as cross-tier interference from MBSs, which
usually transmit with high power. Therefore, interference management
is a critical issue for small cell networks and calls for the joint
optimization of the transmit strategies of the SBSs and MBSs \cite{Hossain14,Andrews14,Andrews12}.

However, the coordination of macrocell and small cell communication
is restricted by the capacity-limited backhaul links between the SBSs
and MBSs \cite{Chen15}. For example, femtocell base stations are
expected to be connected to the core network via Digital Subscriber
Line (DSL) links and there are generally no backhaul links between
them \cite{Andrews12}. Considering the denseness and randomness of
SBS deployment, wireless backhaul methods have also been proposed
for small cell networks \cite{Siddique15}, which, however, have limited
capacities and are vulnerable to dynamic changes in the environment.
Consequently, interference management of small cell networks must
take into account that the information exchange between BSs is limited.
The goal of this paper is to devise distributed optimization methods
for hierarchical small cell networks which can afford only limited
signaling overhead.

Interference management for small cell networks has received much
attention. In \cite{Chandrasekhar09} and \cite{NgoLe12}, the traditional
power control problem was investigated for two-tier code division
multiple access (CDMA) femtocell networks with the aim to meet a signal-to-interference-plus-noise
ratio (SINR) target for every user. Resource allocation for orthogonal
frequency division multiple access (OFDMA) small cell networks was
investigated in a number of works such as \cite{HaLe14,LopezChu14,Abdelnasser15},
which, upon satisfying the QoS constraints protecting the macrocell
communication, tried to maximize the throughput, minimize the transmit
power, or maximize the number of admitted users. In \cite{Ramamonjison15}
and \cite{Elsherif15}, the authors studied energy efficiency maximization
and revenue optimization for two-tier small cell networks, respectively.
Note that the distinguished power control algorithms in \cite{Chandrasekhar09}
and \cite{NgoLe12} were based on the standard function introduced
in \cite{Yates95} that is only applicable for single-variable utilities.
Yet, most resource allocation designs for small cell networks are
based on convex optimization problem formulations or relaxations and
are implemented in a centralized manner. They require the collection
of the channel state information of the entire network and are not
applicable to nonconvex objectives such as sum rate maximization.

An important methodology for distributed interference management is
game theory \cite{YuGinisCioffi02,ScutariFacchinei14,WangScutariPalomar11,WangWangDing15,XuWangWu12,Guruacharya13,ZhangJiang15}.
By formulating the resource competition over interference channels
as a Nash equilibrium problem (NEP), also known as a noncooperative
game, with the aim to achieve an NE, one can obtain completely distributed
transmit strategies \cite{YuGinisCioffi02,ScutariFacchinei14,WangScutariPalomar11,WangWangDing15,XuWangWu12}.
Nevertheless, it is also known that an NE is often socially inefficient
in the sense that either global constraints are violated or the performance
of the whole network is poor. Thus, other game models, such as Stackelberg
game and Nash bargaining, were employed for small cell network design
\cite{Guruacharya13,ZhangJiang15}. However, using these game models
will lead to centralized algorithms which weaken the merit of using
game-based optimization.

In this paper, we study distributed interference management for hierarchical
small cell networks from both the game theoretical and the network
utility maximization (NUM) perspectives. Specifically, we consider
downlink transmission in a two-tier small cell network over multiple
channels and formulate the corresponding power control as two problems,
a noncooperative game and a NUM problem, both with global QoS constraints.
Then, we develop a generalized NEP (GNEP) framework along with various
distributed algorithms and show that the two considered network design
philosophies can be unified under the GNEP framework. The main contributions
of this work include: 
\begin{itemize}
\item We formulate the small cell network design as a GNEP, where the players
are the SBSs and the MBS who compete for the spectral resources by
maximizing their own data rates subject to global QoS constraints
to protect the macrocell communication. 
\item We also formulate the small cell network design as a nonconvex NUM
problem that tries to maximize the sum rate of all BSs subject to
the same global QoS constraints. 
\item To find a generalized Nash equilibrium (GNE) of the formulated GNEP
while satisfying the global QoS constraints, we invoke variational
inequality (VI) theory to analyze the GNEP and characterize the achievable
GNE, referred to as variational equilibrium (VE), based on its existence
and uniqueness properties. 
\item Two alternative distributed algorithms are proposed for finding the
GNE and their convergence properties are analyzed. Both algorithms
only require the macrocell users (MUEs) to broadcast price information. 
\item We further show that the nonconvex NUM problem, although apparently
different from the GNEP, is connected to the GNEP. More precisely,
it is shown that a stationary point of the NUM problem corresponds
to a fixed point of the GNE iteration of a penalized GNEP. 
\item We then propose GNEP-based distributed algorithms to achieve a stationary
solution of the NUM problem at the expense of additional signaling
overhead and complexity. The convergence of the proposed algorithms
is guaranteed by properly chosen algorithm parameters.
\item The developed GNEP framework unifies the game and NUM network designs
as a whole, and is able to scale between them via various GNEP-based
distributed algorithms that offer a tradeoff between performance and
signaling overhead as well as complexity.
\end{itemize}
\textcolor{blue}{\vspace{-0.3cm}}

VI theory \cite{FacchineiPang03,ScutariPalomar10b} is a powerful
tool to analyze and solve noncooperative games and thus has been used
in a number of game-based network designs, such as \cite{ScutariFacchinei14,WangScutariPalomar11,PangScutari08,PangScutariPalomar10,ScutariPalomarFacchineiPang11,WangPeng14,WangWangDing15,StupiaSanguinetti15,BacciBelmega15,ZapponeSanguinetti16}.
In this paper, we utilize VI theory to analyze the GNEP and to find
its GNE. In particular, we show that the considered GNEP can be represented
by a generalized VI (GVI) \cite{FacchineiPang03,FacchineiPang09},
which leads to a distributed pricing mechanism. On the other hand,
GNEP theory \cite{Facchinei07} has not been widely applied to wireless
network optimization, mainly because GNEPs are more complicated than
NEPs (i.e., conventional noncooperative games). Introduced in \cite{PangScutari08}
for studying Gaussian parallel interference channels, GNEPs have been
used to design cognitive radio (CR) networks \cite{PangScutariPalomar10,ScutariPalomarFacchineiPang11,WangPeng14}.
However, small cell networks are different from CR networks as a primary
user in CR is generally passive and not involved in the optimization,
while the MBS in a small cell network is an active resource competitor
and shall be jointly optimized with the SBSs. Hence, the resulting
GNEP for small cell networks is more complicated. 

GNEP-based methods were also proposed in \cite{StupiaSanguinetti15,BacciBelmega15,ZapponeSanguinetti16}
for energy-efficient distributed optimization of heterogeneous small
cell networks. However, global QoS constraints were not considered
in \cite{StupiaSanguinetti15}, while \cite{BacciBelmega15} and \cite{ZapponeSanguinetti16}
relied on a specific analytical form of the best response and the
resulting uniqueness and convergence conditions are difficult to verify.
In this paper, we provide easy-to-check uniqueness and convergence
conditions and our framework can be generalized to other performance
metrics (e.g., mean square errors). Last but not least, to the best
of our knowledge, in the literature there is no work revealing the
connection between GNEPs and common (nonconvex) optimization problems.
In particular, we are the first to connect QoS constrained NUM problems
to GNEPs, further unifying them into a single framework.

This paper is organized as follows. Section II introduces the system
model as well as the game and NUM problem formulations for small cell
networks. In Section III, we exploit VI theory to analyze the formulated
GNEP and identify the achievable GNE. In Section IV, two distributed
algorithms are proposed to achieve the GNE. In Section V, we study
the connection between the NUM problem and the GNEP, and develop GNEP-based
distributed algorithms to solve the NUM problem. Section VI provides
numerical results. Conclusions and extensions are provided in Section
VII.

\textit{Notation:} Upper-case and lower-case boldface letters denote
matrices and vectors, respectively. $\mathbf{I}$ represent the identity
matrix, and $\mathbf{0}$ and $\mathbf{1}$ represent vectors of zeros
and ones, respectively. $\left[\mathbf{A}\right]_{ij}$ denotes the
element in the $i$th row and the $j$th column of matrix $\mathbf{A}$.
$\mathbf{A}\succeq\mathbf{0}$ and $\mathbf{A}\succ\mathbf{0}$ indicate
that $\mathbf{A}$ is a positive semidefinite and a positive definite
matrix, respectively. The operators $\geq$ and $\leq$ are defined
componentwise for vectors and matrices. $\rho\left(\cdot\right)$
and $\sigma_{\max}\left(\cdot\right)$ denote the spectral radius
and the maximum singular value of a matrix, respectively. $\lambda_{\min}\left(\cdot\right)$
denotes the minimum eigenvalue of a Hermitian matrix. $\left\Vert \cdot\right\Vert $
denotes the Euclidean norm of a vector, and $\left\Vert \cdot\right\Vert _{2}$
denotes the spectral norm of a matrix. We define the projection operators
$[x]_{+}\triangleq\max(x,0)$ and $\left[x\right]_{a}^{b}=\max\{a,\min\{b,x\}\}$
for $a\leq b$.

\section{Problem Statement\label{sec:ps}}

\subsection{System Model\label{sec:sm}}

Consider a two-tier hierarchical small cell network consisting of
$M$ SBSs operating in the coverage of an MBS.\footnote{Note that the proposed framework can be readily extended to multiple
MBSs, see Section \ref{sec:cd}.} The SBSs and the MBS share the same downlink resources that are divided
into $N$ channels, which could be time slots in TDMA (Time Division
Multiple Access), frequency bands in FDMA (Frequency Division Multiple
Access), spreading codes in CDMA, or resource blocks in OFDMA. In
the small cells and the macrocell, each downlink channel is assigned
to only one small cell user (SUE) and one macrocell user (MUE), respectively,
such that intra-cell interference does not exist.\footnote{In this paper, we assume that the user association to the BSs is predetermined
and refer the interested reader to \cite{Bethanabhotla16} for user
association optimization.} Hence, the cross-tier interference between the small cells and the
macrocell and the inter-tier interference between the small cells
become the main performance limiting factors \cite{Andrews12,Hossain14}.

For convenience, we index the SBSs as BS $i=1,\ldots,M$ and the MBS
as BS $0$. Denote the channel gain from BS $i$ to the user served
by BS $j$ over channel $n$ by $h_{ij}(n)$. Denote the power allocated
by BS $i$ to channel $n$ by $p_{i}(n)$. Then, the achievable rate
of the user served by BS $i$ over channel $n$ is given by
\begin{equation}
R_{i,n}\triangleq\log\left(1+\frac{h_{ii}(n)p_{i}(n)}{\sigma_{i}(n)+\sum_{j\neq i}h_{ji}(n)p_{j}(n)}\right),\label{sm:rate}
\end{equation}
where $\sigma_{i}(n)$ is the power of the additive white Gaussian
noise at the user served by BS $i$ on channel $n$, and the log function
is the natural logarithm for convenience (i.e., the unit is nats/s/Hz).
The total achievable rate of BS $i$ is 
\begin{equation}
R_{i}(\mathbf{p})=R_{i}(\mathbf{p}_{i},\mathbf{p}_{-i})=\sum_{n=1}^{N}R_{i,n}(\mathbf{p}_{i},\mathbf{p}_{-i}),\label{sm:sum}
\end{equation}
which depends not only on BS $i$'s transmit strategy $\mathbf{p}_{i}\triangleq(p_{i}(n))_{n=1}^{N}$
but also on the other BSs' strategies $\mathbf{p}_{-i}\triangleq(\mathbf{p}_{j})_{j\neq i}$.
The power profile of all BSs' strategies is denoted by $\mathbf{p}\triangleq(\mathbf{p}_{i})_{i=0}^{M}$.
The strategy of each BS shall satisfy the power constraints: 
\begin{align}
\mathbf{p}_{i}\in\mathcal{S}_{i}^{\mathrm{pow}} & \triangleq\left\{ \mathbf{p}_{i}:\sum_{n=1}^{N}p_{i}(n)\leq p_{i}^{\mathrm{sum}},0\leq p_{i}(n)\leq p_{i,n}^{\mathrm{peak}},\forall n\right\} \nonumber \\
 & =\left\{ \mathbf{p}_{i}:\mathbf{1}^{T}\mathbf{p}_{i}\leq p_{i}^{\mathrm{sum}},\mathbf{0}\leq\mathbf{p}_{i}\leq\mathbf{p}_{i}^{\mathrm{peak}}\right\} \label{sm:pow}
\end{align}
where $p_{i}^{\mathrm{sum}}$ is the sum or total power budget of
BS $i$, and $\mathbf{p}_{i}^{\mathrm{peak}}\triangleq(p_{i,n}^{\mathrm{peak}})_{n=1}^{N}$
with $p_{i,n}^{\mathrm{peak}}$ being the peak power budget of BS
$i$ on channel $n$.

Small cells can be considered an enhancement of a macrocell, which
enable the offloading of data traffic from the macrocell and improving
its coverage \cite{Hossain14,Andrews14}. Therefore, macrocell communication
generally has a higher priority and shall be protected from interference
\cite{Hossain14,Andrews14,Andrews12,Chandrasekhar09,NgoLe12,HaLe14,LopezChu14,Abdelnasser15,Ramamonjison15,Elsherif15,Guruacharya13,ZhangJiang15}.
Hence, we impose the global QoS constraints: $R_{0,n}(\mathbf{p})\geq\gamma_{n}$,
$n=1,\ldots,N$, which limit the aggregate interference of all SBSs
on each channel and thus provide a QoS guarantee for each MUE, where
the thresholds $\gamma_{n}$ are chosen such that the QoS constraints
are feasible.\footnote{The QoS constraints are feasible if and only if they are fulfilled
in the absence of interference from the SBSs, i.e., if and only if
$\log(1+\sigma_{0}^{-1}(n)h_{00}(n)p_{0}(n))\geq\gamma_{n}$, $n=1,\ldots,N$,
for some $\mathbf{p}_{0}\in\mathcal{S}_{0}^{\mathrm{pow}}$.} Note that the global QoS constraints depend on both the MBS's and
SBSs' powers, implying that they shall be jointly optimized to meet
the QoS target.

\subsection{Problem Formulation\label{sec:gf}}

In this paper, considering both the game theoretical and the NUM perspectives
for network design, two problem formulations for interference management
in small cell networks are considered. We first formulate the network
design as a noncooperative game, which reflects the competitive nature
of small cell networks, and leads to a completely decentralized optimization.
Specifically, each BS is viewed as a player, i.e., there are $M+1$
players (one MBS and $M$ SBSs). The utility function of each player
$i$ is its rate $R_{i}$, and each player has to meet the QoS and
power constraints. Therefore, the game is formulated as 
\begin{equation}
\mathcal{G}:\;\begin{array}{ll}
\underset{\mathbf{p}_{i}\in\mathcal{S}_{i}^{\mathrm{pow}}}{\mathrm{maximize}} & R_{i}(\mathbf{p}_{i},\mathbf{p}_{-i})\\
\mathrm{subject\,to} & R_{0,n}(\mathbf{p})\geq\gamma_{n},\;\forall n
\end{array}\;i=0,\ldots,M.\label{gf:game}
\end{equation}
We note that $\mathcal{G}$ is different from conventional noncooperative
games or NEPs with decoupled strategy sets \cite{YuGinisCioffi02,ScutariFacchinei14,WangScutariPalomar11}.
Here, the strategy set of BS $i$ is given by 
\begin{equation}
\mathcal{S}_{i}(\mathbf{p}_{-i})=\left\{ \mathbf{p}_{i}\in\mathcal{S}_{i}^{\mathrm{pow}}:R_{0,n}(\mathbf{p}_{i},\mathbf{p}_{-i})\geq\gamma_{n},n=1,\ldots,N\right\} \label{gf:set}
\end{equation}
which clearly depends on other BSs' strategies. Therefore, $\mathcal{G}$
is indeed a generalized Nash equilibrium problem (GNEP) \cite{Facchinei07},
in which the players' strategy sets, in addition to their utility
functions, are coupled. The solution to the GNEP, i.e., the GNE, is
a point $\mathbf{p}^{\star}\triangleq(\mathbf{p}_{i}^{\star})_{i=0}^{M}$
satisfying $R_{i}(\mathbf{p}_{i}^{\star},\mathbf{p}_{-i}^{\star})\geq R_{i}(\mathbf{p}_{i},\mathbf{p}_{-i}^{\star})$,
$\forall\mathbf{p}_{i}\in\mathcal{S}_{i}(\mathbf{p}_{-i}^{\star})$
for $i=0,\ldots,M$. Due to the coupling of the strategy sets, GNEPs
are much more difficult to analyze than NEPs.

Furthermore, we also consider a NUM problem aiming to optimize the
overall system performance under the adopted global QoS constraints.
The most commonly used network utility is the sum rate of the entire
network (i.e., all BSs). Therefore, the QoS constrained NUM problem
is formulated as 
\begin{equation}
\mathcal{P}:\;\begin{array}{ll}
\underset{\mathbf{p}}{\mathrm{maximize}} & \sum_{i=0}^{M}R_{i}(\mathbf{p})\\
\mathrm{subject\,to} & \mathbf{p}_{i}\in\mathcal{S}_{i}^{\mathrm{pow}},\quad i=0,\ldots,M\\
 & R_{0,n}(\mathbf{p})\geq\gamma_{n},\quad n=1,\ldots,N.
\end{array}\label{gf:sm}
\end{equation}
In the literature, this problem is regarded as a difficult problem
due to several unfavorable properties: 1) Problem $\mathcal{P}$ is
NP-hard even in the absence of the QoS constraints \cite{LuoZhang08},
i.e., finding its globally optimal solution requires prohibitive computational
complexity; 2) directly solving problem $\mathcal{P}$ leads to centralized
algorithms that incur significant signaling overheads. Small cell
networks, unlike core networks, usually employ low-cost capacity-limited
backhaul links, which impose strict limitations on the signaling exchange
between BSs (especially SBSs). Therefore, in practice, distributed
methods are often preferred even if they may achieve a suboptimal,
e.g., locally optimal, solution to $\mathcal{P}$.

In the following, we will show that the above two apparently different
network design philosophies can be unified under the GNEP framework.
Specifically, we will first analyze the GNEP $\mathcal{G}$ and provide
different distributed algorithms for finding the GNE of $\mathcal{G}$.
Then, we will investigate the connection between GNEP $\mathcal{G}$
and NUM problem $\mathcal{P}$ and develop GNEP-based distributed
algorithms to achieve a stationary solution of $\mathcal{P}$.

\section{VI Reformulation of the GNEP\label{sec:vr}}

In this section, we analyze the GNEP design of small cell networks.
Due to the coupling of the players' strategy sets, a GNEP is much
more complicated than an NEP and, in its fully general form, is still
deemed intractable \cite{Facchinei07}. Fortunately, the formulated
GNEP for small cell networks enjoys some favorable properties, making
it possible to analyze it and even to find its solution via variational
inequalities (VIs) \cite{FacchineiPang03}. In Appendix \ref{subsec:vi},
we provide a brief introduction to VI theory, while we refer to \cite{FacchineiPang03}
for a more detailed treatment.

Observe that the strategy sets of the BSs, although dependent on each
other, are coupled in a common manner, i.e., by the same QoS constraints
$R_{0,n}(\mathbf{p})\geq\gamma_{n}$ for $n=1,\ldots,N$. Thus, GNEP
$\mathcal{G}$ falls into a class of so-called GNEPs with shared constraints.
Notice further that, for each player $i$, the utility function $R_{i}(\mathbf{p}_{i},\mathbf{p}_{-i})$
is concave in $\mathbf{p}_{i}$, the power constraint set $\mathcal{S}_{i}^{\mathrm{pow}}$
is a convex compact set, and more importantly, the shared QoS constraints
can be rewritten as $\mathbf{g}(\mathbf{p})\leq\mathbf{0}$, where
$\mathbf{g}(\mathbf{p})\triangleq\left(g_{n}(\mathbf{p})\right)_{n=1}^{N}$
and 
\begin{equation}
g_{n}(\mathbf{p})\triangleq\sum_{j=1}^{M}h_{j0}(n)p_{j}(n)+\sigma_{0}(n)-\tilde{h}_{00}(n)p_{0}(n)\label{vr:gn}
\end{equation}
with $\tilde{h}_{00}(n)\triangleq h_{00}(n)/(e^{\gamma_{n}}-1)$.
It is easily seen that $\mathbf{g}(\mathbf{p})$ is jointly convex
(actually linear) in $\mathbf{p}$ and the shared QoS constraints
are convex. Consequently, GNEP $\mathcal{G}$ can be further classified
as a GNEP with jointly convex shared constraints or shortly a jointly
convex GNEP \cite{Facchinei07}.

A jointly convex GNEP, though simpler than its general form, is still
a difficult problem and finding all its solutions or GNEs is still
intractable. However, we are able to characterize a class of GNEs,
called variational equilibria (VEs) \cite{Facchinei07} (also known
as normalized equilibria \cite{Rosen65}), through VI theory. For
this purpose, we introduce $\mathcal{S}^{\mathrm{pow}}\triangleq\prod_{i=0}^{M}\mathcal{S}_{i}^{\mathrm{pow}}$
and 
\begin{align}
\mathcal{S} & \triangleq\left\{ \mathbf{p}:\mathbf{g}(\mathbf{p})\leq\mathbf{0},\;\mathbf{p}_{i}\in\mathcal{S}_{i}^{\mathrm{pow}},\;\forall i\right\} \nonumber \\
 & =\left\{ \mathbf{p}:\mathbf{g}(\mathbf{p})\leq\mathbf{0},\;\mathbf{p}\in\mathcal{S}^{\mathrm{pow}}\right\} .\label{vr:sp}
\end{align}
It is easily seen that $\mathcal{S}_{i}(\mathbf{p}_{-i})$ in (\ref{gf:set})
is a slice of $\mathcal{S}$, i.e., $\mathcal{S}_{i}(\mathbf{p}_{-i})=\left\{ \mathbf{p}_{i}:(\mathbf{p}_{i},\mathbf{p}_{-i})\in\mathcal{S}\right\} $.
Let $\mathbf{f}(\mathbf{p})\triangleq\left(\mathbf{f}_{i}(\mathbf{p})\right)_{i=0}^{M}$
and 
\begin{align}
\mathbf{f}_{i}(\mathbf{p}) & \triangleq-\nabla_{\mathbf{p}_{i}}R_{i}(\mathbf{p})=\left(-\partial R_{i,n}/\partial p_{i}(n)\right)_{n=1}^{N}\nonumber \\
 & =\left(-h_{ii}(n)I_{i,n}^{-1}(\mathbf{p})\right)_{n=1}^{N}\label{vr:fi}
\end{align}
where $I_{i,n}(\mathbf{p})\triangleq\sigma_{i}(n)+\sum_{j=0}^{M}h_{ji}(n)p_{j}(n)$.
Then, GNEP $\mathcal{G}$ is linked to the following VI.
\begin{lem}
\label{lem:vg}A solution of $\mathrm{VI}(\mathcal{S},\mathbf{f})$,
i.e., a vector $\mathbf{p}^{\star}$ such that $(\mathbf{p}-\mathbf{p}^{\star})^{T}\mathbf{f}(\mathbf{p}^{\star})\geq0$,
$\forall\mathbf{p}\in\mathcal{S}$, is also a GNE of $\mathcal{G}$.
\end{lem}
\begin{proof}
Lemma \ref{lem:vg} is proved by comparing the optimality conditions
of the GNE of $\mathcal{G}$ and the solution of $\mathrm{VI}(\mathcal{S},\mathbf{f})$
\cite{Facchinei07}.
\end{proof}
Lemma \ref{lem:vg} indicates that a subset of the GNEs, i.e., the
VEs, of $\mathcal{G}$ can be characterized by $\mathrm{VI}(\mathcal{S},\mathbf{f})$.\footnote{Note that GNEP $\mathcal{G}$ can also be transformed into a quasi-variational
inequality (QVI) \cite{PangFukushima05,StupiaSanguinetti15}.} VEs are a class of solutions of jointly convex GNEPs that can be
found efficiently. Thus, it is reasonable to focus on the VE of $\mathcal{G}$.
Lemma \ref{lem:vg} also enables us to investigate the existence and
even uniqueness of a GNE by invoking VI theory. Particularly, a unique
VI solution is implied by the uniformly P property, which means that
if $\mathbf{f}$ is a uniformly P function (see Appendix \ref{subsec:vi}),
then $\mathrm{VI}(\mathcal{S},\mathbf{f})$ has a unique solution
(thus a unique VE). However, in practice, it is hard and inconvenient
to verify the uniformly P property of $\mathbf{f}$ by its definition.
Hence, in the following, we provide an easy-to-check condition for
whether $\mathcal{G}$ has a unique VE.
\begin{prop}
\label{pro:unq}GNEP $\mathcal{G}$ always admits a GNE and has a
unique VE if $\Psi$ is a P-matrix, where $\Psi\in\mathbb{R}^{(M+1)\times(M+1)}$
is defined as\textup{
\begin{equation}
\left[\mathbf{\Psi}\right]_{ij}\triangleq\begin{cases}
\min_{n}\frac{h_{ii}^{2}(n)}{\left(\sigma_{i}(n)+\sum_{l=0}^{M}h_{li}(n)p_{l,n}^{\mathrm{max}}\right)^{2}}, & \mbox{if }i=j\\
-\max_{n}\frac{h_{ii}(n)h_{ji}(n)}{\sigma_{i}^{2}(n)}, & \mbox{if }i\neq j
\end{cases}\label{vr:ms}
\end{equation}
with $p_{l,n}^{\mathrm{max}}\triangleq\min\{p_{l}^{\mathrm{sum}},p_{l,n}^{\mathrm{peak}}\}$
for $l=0,\ldots,M$ and $n=1,\ldots,N$.}
\end{prop}
\begin{proof}
See Appendix \ref{subsec:unq}.
\end{proof}
Accompanied by Proposition \ref{pro:unq} is the following result.
\begin{lem}
\label{lem:sfe}$\Psi$ is a P-matrix if and only if $\rho(\Phi)<1$,
where $\Phi\in\mathbb{R}^{(M+1)\times(M+1)}$ is defined as
\begin{equation}
\left[\mathbf{\Phi}\right]_{ij}\triangleq\begin{cases}
0, & \mbox{if }i=j\\
-\frac{\left[\mathbf{\Psi}\right]_{ij}}{\left[\mathbf{\Psi}\right]_{ii}}, & \mbox{if }i\neq j.
\end{cases}\label{vr:mf}
\end{equation}
\end{lem}
\begin{proof}
Lemma \ref{lem:sfe} follows from Lemma \ref{lem:kmx} in Appendix
\ref{subsec:vi}.
\end{proof}
From Proposition \ref{pro:unq} and Lemma \ref{lem:sfe}, there is
a unique VE if the real matrix $\mathbf{\Psi}$ is a P-matrix (see
Appendix \ref{subsec:vi}) or $\rho(\mathbf{\Phi})<1$. By definition
of the P-matrix, a positive definite matrix is also a P-matrix, so
one can also check the positive definiteness of $\mathbf{\Psi}$,
which is implied by the strict diagonal dominance, i.e., $\left[\mathbf{\Psi}\right]_{i}>\sum_{j\neq i}\left[\left|\mathbf{\Psi}\right|\right]_{ij}$
and $\left[\mathbf{\Psi}\right]_{j}>\sum_{i\neq j}\left[\left|\mathbf{\Psi}\right|\right]_{ij}$
for $i,j=0,\ldots,M$. This can also be intuitively interpreted as
the information signals of each BS being stronger than the corresponding
interference \cite{ScutariFacchinei14,WangScutariPalomar11}.

Now, we investigate how to obtain a VE of $\mathcal{G}$. A natural
way to compute the VE, as pointed out in Lemma \ref{lem:vg}, is to
directly solve $\mathrm{VI}(\mathcal{S},\mathbf{f})$. Considering
that function $\mathbf{f}$ and set $\mathcal{S}$ are coupled by
all BSs, this approach, however, leads to a centralized algorithm,
which contradicts our desired goal of distributed optimization. For
a decentralized design, we introduce a generalized VI (GVI) (see Appendix
\ref{subsec:vi} or \cite{FacchineiPang03,FacchineiPang09}) based
on GNEP $\mathcal{G}$. Specifically, consider the following noncooperative
game or NEP: 
\begin{equation}
\mathcal{G}_{\boldsymbol{\mu}}:\;\begin{array}{ll}
\underset{\mathbf{p}_{i}}{\mathrm{maximize}} & R_{i}(\mathbf{p}_{i},\mathbf{p}_{-i})-\boldsymbol{\mu}^{T}\mathbf{g}(\mathbf{p})\\
\mathrm{subject\,to} & \mathbf{p}_{i}\in\mathcal{S}_{i}^{\mathrm{pow}}
\end{array}\;i=0,\ldots,M\label{vr:gu}
\end{equation}
where $\boldsymbol{\mu}\triangleq\left(\mu_{n}\right)_{n=1}^{N}\geq\mathbf{0}$
is a given nonnegative vector. This is a conventional NEP with decoupled
strategy sets. We denote its NE by $\mathbf{p}^{\star}(\boldsymbol{\mu})$,
which is a function of $\boldsymbol{\mu}$, as is $\mathbf{g}(\mathbf{p}^{\star}(\boldsymbol{\mu}))$.
Considering that there may be multiple NEs, the values of $\mathbf{g}(\mathbf{p}^{\star}(\boldsymbol{\mu}))$
could be a set $\left\{ \mathbf{g}(\mathbf{p}^{\star}(\boldsymbol{\mu}))\right\} $.
Thus, we define a point-to-set map $\mathcal{Q}(\boldsymbol{\mu}):\mathbb{R}_{+}^{N}\rightarrow\left\{ -\mathbf{g}(\mathbf{p}^{\star}(\boldsymbol{\mu}))\right\} $.
Then, we introduce $\mathrm{GVI}(\mathbb{R}_{+}^{N},\mathcal{Q})$,
whose solution is a vector $\boldsymbol{\mu}^{\star}$ such that $\mathcal{G}_{\boldsymbol{\mu}^{\star}}$
admits an NE $\mathbf{p}^{\star}(\boldsymbol{\mu}^{\star})$ and 
\begin{equation}
(\boldsymbol{\mu}-\boldsymbol{\mu}^{\star})^{T}\mathbf{g}(\mathbf{p}^{\star}(\boldsymbol{\mu}^{\star}))\leq0,\quad\forall\boldsymbol{\mu}\geq\mathbf{0}.\label{vr:lh}
\end{equation}
The relation between $\mathrm{GVI}(\mathbb{R}_{+}^{N},\mathcal{Q})$
and $\mathrm{VI}(\mathcal{S},\mathbf{f})$ (and also $\mathcal{G}$)
is given in the following theorem.
\begin{thm}
\label{thm:gvi}If $\boldsymbol{\mu}^{\star}$ is a solution of $\mathrm{GVI}(\mathbb{R}_{+}^{N},\mathcal{Q})$
with $\mathbf{p}^{\star}(\boldsymbol{\mu}^{\star})$ being an NE of
$\mathcal{G}_{\boldsymbol{\mu}^{\star}}$, then $\mathbf{p}^{\star}(\boldsymbol{\mu}^{\star})$
is a solution of $\mathrm{VI}(\mathcal{S},\mathbf{f})$ with $\boldsymbol{\mu}^{\star}$
being the Lagrange multiplier associated with the constraint $\mathbf{g}(\mathbf{p})\leq\mathbf{0}$.
Conversely, if $\mathbf{p}^{\star}$ is a solution of $\mathrm{VI}(\mathcal{S},\mathbf{f})$
and $\boldsymbol{\mu}^{\star}$ is the Lagrange multiplier associated
with the constraint $\mathbf{g}(\mathbf{p})\leq\mathbf{0}$, then
$\boldsymbol{\mu}^{\star}$ is a solution of $\mathrm{GVI}(\mathbb{R}_{+}^{N},\mathcal{Q})$
and $\mathbf{p}^{\star}$ is an NE of $\mathcal{G}_{\boldsymbol{\mu}^{\star}}$.
\end{thm}
\begin{proof}
The proof is based on comparing the Karush-Kuhn-Tucker (KKT) conditions
of $\mathrm{GVI}(\mathbb{R}_{+}^{N},\mathcal{Q})$ and $\mathrm{VI}(\mathcal{S},\mathbf{f})$.
Due to the space limitation, we refer the interested reader to \cite{FacchineiPang03,WangPeng14}
for details..
\end{proof}
Theorem \ref{thm:gvi} establishes the equivalence between $\mathrm{VI}(\mathcal{S},\mathbf{f})$
and $\mathrm{GVI}(\mathbb{R}_{+}^{N},\mathcal{Q})$ and enables us
to obtain a GNE of $\mathcal{G}$ by solving $\mathrm{GVI}(\mathbb{R}_{+}^{N},\mathcal{Q})$
instead. The benefit of this approach is its amenability to pricing
mechanisms, which facilitate the development of distributed algorithms
for computing the VE of $\mathcal{G}$. Specifically, the nonnegative
vector $\boldsymbol{\mu}$ can be regarded as the price of violating
the QoS constraint $\mathbf{g}(\mathbf{p})\leq\mathbf{0}$, and the
term $\boldsymbol{\mu}^{T}\mathbf{g}(\mathbf{p})$ is the cost paid
by all BSs. Given the price $\boldsymbol{\mu}$, the BSs (including
both the MBS and the SBSs) will compete and play NEP $\mathcal{G}_{\boldsymbol{\mu}}$
to reach an NE $\mathbf{p}^{\star}(\boldsymbol{\mu})$. The task of
$\mathrm{GVI}(\mathbb{R}_{+}^{N},\mathcal{Q})$ is to choose an appropriate
$\boldsymbol{\mu}^{\star}$ so that at this point NE $\mathbf{p}^{\star}(\boldsymbol{\mu}^{\star})$
is also a VE of GNEP $\mathcal{G}$. Consequently, the difficult problem
of finding a GNE of $\mathcal{G}$ is decomposed into two subproblems:
1) How to solve NEP $\mathcal{G}_{\boldsymbol{\mu}}$ for a given
price, and 2) how to choose the price $\boldsymbol{\mu}$. In the
next section, we will show that these two subproblems can both be
addressed distributively.

\section{Distributed Computation of GNE\label{sec:dcg}}

\subsection{Distributed Pricing Algorithm\label{sec:dm}}

In this subsection, we will establish a distributed pricing mechanism
to achieve the VE of $\mathcal{G}$ by solving the two subproblems
mentioned above. We first address the subproblem of how to solve NEP
$\mathcal{G}_{\boldsymbol{\mu}}$ for a given price $\boldsymbol{\mu}$.
Our focus is on obtaining the NE via the best response algorithm that
only uses local information. To this end, we shall first investigate
the existence and uniqueness properties of the NE of $\mathcal{G}_{\boldsymbol{\mu}}$.
This can be done by linking $\mathcal{G}_{\boldsymbol{\mu}}$ to a
VI.

Let us introduce $\mathcal{S}^{\mathrm{pow}}\triangleq\prod_{i=0}^{M}\mathcal{S}_{i}^{\mathrm{pow}}$
and $\mathbf{f}_{\boldsymbol{\mu}}(\mathbf{p})\triangleq\left(\mathbf{f}_{\boldsymbol{\mu},i}(\mathbf{p})\right)_{i=0}^{M}$
with 
\begin{align}
\mathbf{f}_{\boldsymbol{\mu},i}(\mathbf{p}) & \triangleq-\nabla_{\mathbf{p}_{i}}R_{i}(\mathbf{p})+\sum_{n=1}^{N}\mu_{n}\nabla_{\mathbf{p}_{i}}g_{n}(\mathbf{p})\nonumber \\
 & =\begin{cases}
(\mathbf{f}_{0}(\mathbf{p})-\tilde{h}_{00}(n)\mu_{n})_{n=1}^{N}, & i=0\\
(\mathbf{f}_{i}(\mathbf{p})+h_{i0}(n)\mu_{n})_{n=1}^{N}, & i=1,\ldots,M
\end{cases}\label{dm:fx}
\end{align}
where $\nabla_{\mathbf{p}_{i}}R_{i}(\mathbf{p})$ and $\nabla_{\mathbf{p}_{i}}g_{n}(\mathbf{p})$
are the partial derivatives of $R_{i}(\mathbf{p})$ and $g_{n}(\mathbf{p})$
with respect to $\mathbf{p}_{i}$, respectively. Then, NEP $\mathcal{G}_{\boldsymbol{\mu}}$
is equivalent to the following VI based on $\mathcal{S}^{\mathrm{pow}}$
and $\mathbf{f}_{\boldsymbol{\mu}}(\mathbf{p})$.
\begin{lem}
\label{lem:uvi}Given $\boldsymbol{\mu}\geq\mathbf{0}$, $\mathcal{G}_{\boldsymbol{\mu}}$
is equivalent to $\mathrm{VI}(\mathcal{S}^{\mathrm{pow}},\mathbf{f}_{\boldsymbol{\mu}})$,
i.e., $\mathbf{p}^{\star}$ is an NE of $\mathcal{G}_{\boldsymbol{\mu}}$
if and only if $\mathbf{p}^{\star}$ satisfies $\left(\mathbf{p}-\mathbf{p}^{\star}\right)^{T}\mathbf{f}_{\boldsymbol{\mu}}\left(\mathbf{p}^{\star}\right)\geq0$,
$\forall\mathbf{p}\in\mathcal{S}^{\mathrm{pow}}$.
\end{lem}
\begin{proof}
Lemma \ref{lem:uvi} is proved by comparing the optimality conditions
of the NE of $\mathcal{G}_{\boldsymbol{\mu}}$ and the solution of
$\mathrm{VI}(\mathcal{S}^{\mathrm{pow}},\mathbf{f}_{\boldsymbol{\mu}})$
\cite{FacchineiPang03}.
\end{proof}
With the help of Lemma \ref{lem:uvi}, we can now analyze $\mathcal{G}_{\boldsymbol{\mu}}$
via $\mathrm{VI}(\mathcal{S}^{\mathrm{pow}},\mathbf{f}_{\boldsymbol{\mu}})$
using established facts from VI theory. The existence of an NE of
$\mathcal{G}_{\boldsymbol{\mu}}$ is always guaranteed \cite{Rosen65},
since for each player $i$ the utility in (\ref{vr:gu}) is concave
in $\mathbf{p}_{i}$ and the strategy set $\mathcal{S}_{i}^{\mathrm{pow}}$
is convex and compact. The uniqueness of the solution of $\mathrm{VI}(\mathcal{S}^{\mathrm{pow}},\mathbf{f}_{\boldsymbol{\mu}})$
is implied by the uniformly P property of $\mathbf{f}_{\boldsymbol{\mu}}$.
Similar to Proposition \ref{pro:unq}, we are also able to provide
a sufficient condition for a unique NE. 
\begin{prop}
\label{pro:ung}Given $\boldsymbol{\mu}\geq\mathbf{0}$, $\mathrm{VI}(\mathcal{S}^{\mathrm{pow}},\mathbf{f}_{\boldsymbol{\mu}})$
($\mathcal{G}_{\boldsymbol{\mu}}$) has a unique solution (NE) if
$\Psi$ is a P-matrix or $\rho(\mathbf{\Phi})<1$. 
\end{prop}
\begin{proof}
Since the term $\sum_{n=1}^{N}\mu_{n}\nabla_{\mathbf{p}_{i}}g_{n}(\mathbf{p})$
in $\mathbf{f}_{\boldsymbol{\mu},i}(\mathbf{p})$ is a constant, the
uniformly P property of $\mathbf{f}_{\boldsymbol{\mu}}$ is implied
by that of $\mathbf{f}(\mathbf{p})$, which has been proved in Proposition
\ref{pro:unq}.
\end{proof}
Interestingly, Propositions \ref{pro:unq} and \ref{pro:ung} provide
the same uniqueness condition. Therefore, if GNEP $\mathcal{G}$ admits
a unique VE, then NEP $\mathcal{G}_{\boldsymbol{\mu}}$ also has a
unique NE, regardless of price $\boldsymbol{\mu}$. As mentioned above,
the condition of $\Psi$ being a P-matrix or $\rho(\mathbf{\Phi})<1$
can be understood as the interference in the small cell network being
not too large.

Now, we consider the decentralized computation of the NE of $\mathcal{G}_{\boldsymbol{\mu}}$
via the best response algorithm, i.e., each BS aims to maximize its
own utility by solving the problem in (\ref{vr:gu}). More exactly,
in each iteration, the MBS (BS 0) and SBSs (BS $i=1,\ldots,M$) shall
solve the following equivalent problems, respectively, 
\begin{alignat}{1}
 & \underset{\mathbf{p}_{0}\in\mathcal{S}_{0}^{\mathrm{pow}}}{\mathrm{maximize}}\;R_{i}(\mathbf{p}_{0},\mathbf{p}_{-0})+\sum_{n=1}^{N}\mu_{n}\tilde{h}_{00}(n)p_{0}(n)\label{dm:mbs}\\
 & \underset{\mathbf{p}_{i}\in\mathcal{S}_{i}^{\mathrm{pow}}}{\mathrm{maximize}}\;R_{i}(\mathbf{p}_{i},\mathbf{p}_{-i})-\sum_{n=1}^{N}\mu_{n}h_{i0}(n)p_{i}(n).\label{dm:sbs}
\end{alignat}
We are able to find the closed-form solution to (\ref{dm:mbs}) and
(\ref{dm:sbs}): 
\begin{multline}
p_{i}^{\star}(n)=V_{i,n}(\mathbf{p}_{-i},\boldsymbol{\mu})\triangleq\\
\begin{cases}
\left[\frac{1}{\lambda_{0}-\mu_{n}\tilde{h}_{00}(n)}-\frac{I_{0,n}(\mathbf{p}_{-0})}{h_{00}(n)}\right]_{0}^{p_{0,n}^{\mathrm{peak}}}, & i=0\\
\left[\frac{1}{\lambda_{i}+\mu_{n}h_{i0}(n)}-\frac{I_{i,n}(\mathbf{p}_{-i})}{h_{ii}(n)}\right]_{0}^{p_{i,n}^{\mathrm{peak}}}, & i=1,\ldots,M,
\end{cases}\label{dm:pi}
\end{multline}
where $I_{i,n}(\mathbf{p}_{-i})\triangleq\sigma_{i}(n)+\sum_{j\neq i}h_{ji}(n)p_{j}(n)$,
and $\lambda_{i}$ is the minimum value such that $\sum_{n=1}^{N}p_{i}^{\star}(n)\leq p_{i}^{\mathrm{sum}}$
for $\forall i$. By defining $\mathbf{V}_{i}(\mathbf{p}_{-i},\boldsymbol{\mu})\triangleq\left(V_{i,n}(\mathbf{p}_{-i},\boldsymbol{\mu})\right)_{n=1}^{N}$,
the best response of each BS can be compactly expressed as $\mathbf{p}_{i}^{\star}=\mathbf{V}_{i}(\mathbf{p}_{-i},\boldsymbol{\mu})$
for $i=0,\ldots,M$. The best response algorithm is formally stated
in Algorithm 1, where $\mathbf{p}^{t}\triangleq(\mathbf{p}_{i}^{t})_{i=0}^{M}$
represents the strategy profile generated in iteration $t$. The convergence
of Algorithm 1 is characterized in Proposition \ref{pro:bes}.

\begin{algorithm}[tbph]
\caption{\textbf{: Distributed Best Response Algorithm for }$\mathcal{G}_{\boldsymbol{\mu}}$}

1: Set the initial point $\mathbf{p}^{0}$, precision $\epsilon$,
and $t=0$;

2: Update $\mathbf{p}_{i}^{t+1}=\mathbf{V}_{i}\left(\mathbf{p}_{-i}^{t},\boldsymbol{\mu}\right)$
for $i=0,\ldots,M$;

3: $t=t+1$;

4: If $\left\Vert \mathbf{p}^{t}-\mathbf{p}^{t-1}\right\Vert \leq\epsilon$
stop, otherwise go to step 2.
\end{algorithm}

\begin{prop}
\label{pro:bes}The sequence $\left\{ \mathbf{p}^{t}\right\} _{t=0}^{\infty}$
generated by Algorithm 1 converges to the unique NE of $\mathcal{G}_{\boldsymbol{\mu}}$,
provided that $\mathbf{\Psi}$ is a P-matrix or $\rho(\mathbf{\Phi})<1$.
\end{prop}
\begin{proof}
See Appendix \ref{subsec:bes}.
\end{proof}
Next, we consider the second subproblem of how to choose price $\boldsymbol{\mu}$
to solve $\mathrm{GVI}(\mathbb{R}_{+}^{N},\mathcal{Q})$ and obtain
the VE of $\mathcal{G}$. For this purpose, we shall investigate how
the global QoS constraint $\mathbf{g}(\mathbf{p})\leq\mathbf{0}$
is related to price $\boldsymbol{\mu}$. Since NE $\mathbf{p}^{\star}(\boldsymbol{\mu})$
of $\mathcal{G}_{\boldsymbol{\mu}}$ is a function of $\boldsymbol{\mu}$,
$\mathbf{g}(\mathbf{p}^{\star}(\boldsymbol{\mu}))$, or for short
$\mathbf{g}(\boldsymbol{\mu})$, is also a function of $\boldsymbol{\mu}$
but through a rather complicated relation. Given that $\mathbf{\Psi}$
is a P-matrix, $\mathbf{g}(\boldsymbol{\mu})$ is unique and thus
multifunction $\mathcal{Q}(\boldsymbol{\mu}):\mathbb{R}_{+}^{N}\rightarrow\left\{ -\mathbf{g}(\boldsymbol{\mu})\right\} $
reduces to a single-valued function $\mathcal{Q}(\boldsymbol{\mu})=-\mathbf{g}(\boldsymbol{\mu})$,
so $\mathrm{GVI}(\mathbb{R}_{+}^{N},\mathcal{Q})$ becomes $\mathrm{VI}(\mathbb{R}_{+}^{N},-\mathbf{g}(\boldsymbol{\mu}))$.
Interestingly, $\mathbf{g}(\boldsymbol{\mu})$ has the following property.
\begin{lem}
\label{lem:coc}(\cite{ScutariPalomarFacchineiPang11}) Given $\mathbf{\Psi}\succ\mathbf{0}$,
$-\mathbf{g}(\boldsymbol{\mu})$ is co-coercive in $\boldsymbol{\mu}$,
i.e., there exists a constant $c_{coc}$ such that $\left(\boldsymbol{\mu}_{1}-\boldsymbol{\mu}_{2}\right)^{T}\left(\mathbf{g}(\boldsymbol{\mu}_{2})-\mathbf{g}(\boldsymbol{\mu}_{1})\right)\geq c_{\mathrm{coc}}\left\Vert \mathbf{g}(\boldsymbol{\mu}_{2})-\mathbf{g}(\boldsymbol{\mu}_{1})\right\Vert ^{2}$,
$\forall\boldsymbol{\mu}_{1},\boldsymbol{\mu}_{2}\in\mathbb{R}_{+}^{N}$.
\end{lem}
Since a positive definite matrix is also a P-matrix, Lemma \ref{lem:coc}
is consistent with Propositions \ref{pro:ung} and \ref{pro:bes}.
Co-coercivity plays an important role in VIs similar to convexity
in optimization. The co-coercivity of $-\mathbf{g}(\boldsymbol{\mu})$
guarantees that there exists a solution of $\mathrm{VI}(\mathbb{R}_{+}^{N},-\mathbf{g}(\boldsymbol{\mu}))$
and thus a GNE (VE) of $\mathcal{G}$. Moreover, this favorable property
enables us to devise a distributed price updating algorithm, i.e.,
Algorithm 2, to find the solution of $\mathrm{VI}(\mathbb{R}_{+}^{N},-\mathbf{g}(\boldsymbol{\mu}))$
or the VE of $\mathcal{G}$.

\begin{algorithm}[tbph]
\caption{\textbf{: Distributed Pricing Algorithm for }$\mathcal{G}$}

1: Set the initial point $\boldsymbol{\mu}^{0}$, precision $\epsilon$,
and $k=0$;

2: Compute the NE $\mathbf{p}^{\star}(\boldsymbol{\mu}^{k})$ of $\mathcal{G}_{\boldsymbol{\mu}^{k}}$
via Algorithm 1;

3: Update the price as $\boldsymbol{\mu}^{k+1}=\left[\boldsymbol{\mu}^{k}-\eta_{k}\mathbf{g}(\boldsymbol{\mu}^{k})\right]_{+}$;

4: $k=k+1$;

5: If $\left\Vert \boldsymbol{\mu}^{k}-\boldsymbol{\mu}^{k-1}\right\Vert \leq\epsilon$
stop, otherwise go to step 2.
\end{algorithm}

Algorithm 2 contains two loops, where the outer loop is to update
the price vector $\boldsymbol{\mu}$, and the inner loop invokes Algorithm
1 to obtain the NE of $\mathcal{G}_{\boldsymbol{\mu}}$. In Algorithm
2, $\eta_{k}$ is a step size, which could be a constant or vary in
each iteration. The co-coercivity of $-\mathbf{g}(\boldsymbol{\mu})$
guarantees that Algorithm 2 converges to the solution of $\mathrm{VI}(\mathbb{R}_{+}^{N},-\mathbf{g}(\boldsymbol{\mu}))$
with a properly chosen step size. Consequently, we have the following
result.
\begin{thm}
\label{thm:pcv}Given $\mathbf{\Psi}\succ\mathbf{0}$ and $0<\eta_{k}<2c_{\mathrm{coc}}$
for $\forall k$, the sequence $\{\boldsymbol{\mu}^{k}\}_{k=0}^{\infty}$
generated by Algorithm 2 converges to a solution $\boldsymbol{\mu}^{\star}$
of $\mathrm{VI}(\mathbb{R}_{+}^{N},-\mathbf{g}(\boldsymbol{\mu}))$\textup{
and }$\mathbf{p}^{\star}(\boldsymbol{\mu}^{\star})$\textup{ is a
GNE (VE) of }$\mathcal{G}$.
\end{thm}
\begin{proof}
Theorem \ref{thm:pcv} follows from Lemma \ref{lem:coc} and \cite[Th. 12.1.8]{FacchineiPang03}.
\end{proof}
The implementation of Algorithms 1 and 2 in small cell networks leads
to a distributed pricing mechanism. Specifically, the MUEs are responsible
for updating the price according to Algorithm 2. For this purpose,
the MUEs need to know $g_{n}(\mathbf{p}(\boldsymbol{\mu}))$, which,
from (\ref{vr:gn}), contains the aggregate interference (plus noise)
$\sum_{j=1}^{M}h_{j0}(n)p_{j}(n)+\sigma_{0}(n)$ from the small cells
and the (normalized) received power $\tilde{h}_{00}(n)p_{0}(n)$ from
the MBS, both of which can be locally measured by each MUE. Then,
the MUE using channel $n$ broadcasts its price $\mu_{n}$ for $n=1,\ldots,N$.
With the given price, all BSs (the MBS and the SBSs) distributively
compute the NE of $\mathcal{G}_{\boldsymbol{\mu}}$ via Algorithm
1. In each iteration of Algorithm 1, according to the best response
in (\ref{dm:pi}), each BS $i$ needs to know the direct channel $h_{ii}(n)$
and the aggregate interference $I_{i,n}(\mathbf{p}_{-i})$ from the
other cells, while each SBS $i$ also needs to know the term $\mu_{n}h_{i0}(n)$.
It is easily seen that $h_{ii}(n)$ and $I_{i,n}(\mathbf{p}_{-i})$
can be locally estimated or measured by the user served by BS $i$
and be fed back to the BS. Since the price $\mu_{n}$ is broadcast
by the MUE on channel $n$, the term $\mu_{n}h_{i0}(n)$ can also
be locally measured by SBS $i$ by exploiting the reciprocity of the
channel between SBS $i$ and the MUE. Consequently, the whole pricing
mechanism only needs the MUEs to broadcast the price information. 

\subsection{Distributed Proximal Algorithm\label{sec:da}}

The above distributed pricing algorithm includes two time scales,
a faster one for power updating and a slower one for price updating.
Naturally, one may wonder, in the hope of accelerating the convergence
speed, if price and power can be updated simultaneously. The answer
is, however, complicated. In this subsection, we show simultaneous
updating of price and power is possible, but a two-loop structure
is still needed to guarantee convergence.

Inspired by the NEP with given price, a possible option is to incorporate
the price into the power update by viewing the price updater as an
additional player \cite{ScutariPalomarFacchineiPang11}. According
to $\mathrm{GVI}(\mathbb{R}_{+}^{N},\mathcal{Q})$, the optimal price
$\boldsymbol{\mu}^{\star}$ is chosen to satisfy $(\boldsymbol{\mu}-\boldsymbol{\mu}^{\star})^{T}\mathbf{g}(\mathbf{p}^{\star})\leq0$,
$\forall\boldsymbol{\mu}\geq\mathbf{0}$ with $\mathbf{p}^{\star}$
being the NE of $\mathcal{G}_{\boldsymbol{\mu}^{\star}}$, which is
exactly the first-order optimality condition \cite{BoydVan04} of
the maximization problem $\mathrm{maximize}_{\boldsymbol{\mu}\geq\mathbf{0}}\;\boldsymbol{\mu}^{T}\mathbf{g}(\mathbf{p}^{\star})$.
Therefore, we are able to incorporate the price into $\mathcal{G}_{\boldsymbol{\mu}}$
and formulate a new NEP with $M+2$ players, where the first $M$+1
players are the MBS and the SBSs who optimize the transmit power according
to 
\begin{equation}
\underset{\mathbf{p}_{i}\in\mathcal{S}_{i}^{\mathrm{pow}}}{\mathrm{maximize}}\;R_{i}(\mathbf{p}_{i},\mathbf{p}_{-i})-\boldsymbol{\mu}^{T}\mathbf{g}(\mathbf{p})\quad i=0,\ldots,M,\label{da:mp}
\end{equation}
and the $(M+2)$th player is the price updater who optimizes the price
according to 
\begin{equation}
\underset{\boldsymbol{\mu}\geq\mathbf{0}}{\mathrm{maximize}}\;\boldsymbol{\mu}^{T}\mathbf{g}(\mathbf{p}).\label{da:mu}
\end{equation}
We refer to the combination of (\ref{da:mp}) and (\ref{da:mu}) as
NEP $\bar{\mathcal{G}}$. Then, NEP $\bar{\mathcal{G}}$ is linked
to GNEP $\mathcal{G}$ and $\mathrm{VI}(\mathcal{S},\mathbf{f})$
as follows \cite{ScutariPalomarFacchineiPang11}.
\begin{lem}
\label{lem:pue}$(\mathbf{p}^{\star},\boldsymbol{\mu}^{\star})$ is
an NE of $\bar{\mathcal{G}}$ if and only if $\mathbf{p}^{\star}$
is a solution to $\mathrm{VI}(\mathcal{S},\mathbf{f})$ and thus a
VE of $\mathcal{G}$, and $\boldsymbol{\mu}^{\star}$ is the Lagrange
multiplier associated with the constraint $\mathbf{g}(\mathbf{p})\leq\mathbf{0}$.
\end{lem}
It is now clear that the VE of $\mathcal{G}$ can be obtained by finding
the NE of $\bar{\mathcal{G}}$, which implies that power and price
can be updated simultaneously. Herein, it is natural to exploit the
best response algorithm (e.g., Algorithm 1) to compute the NE of $\bar{\mathcal{G}}$.
However, directly applying the best response algorithm to $\bar{\mathcal{G}}$
may lead to divergence. Indeed, if one constructs an $(M+2)\times(M+2)$
matrix $\bar{\Psi}$ similar to $\Psi$ in (\ref{vr:ms}), it will
result in $[\bar{\Psi}]_{(M+2)(M+2)}=0$, implying that $\bar{\Psi}$
cannot be a P-matrix or a positive definite matrix, so convergence
is not guaranteed. The principal reason is that the utility function
in (\ref{da:mu}) is neither strictly concave nor strongly concave
in $\boldsymbol{\mu}$.

To overcome this difficulty, we reformulate NEP $\bar{\mathcal{G}}$
into a VI. Introduce $\bar{\mathbf{p}}\triangleq(\mathbf{p},\boldsymbol{\mu})$,
$\bar{\mathcal{S}}\triangleq\prod_{i=0}^{M}\mathcal{S}_{i}^{\mathrm{pow}}\times\mathbb{R}_{+}^{N}$,
and $\bar{\mathbf{f}}(\bar{\mathbf{p}})\triangleq\left(\left(\mathbf{f}_{\boldsymbol{\mu},i}(\mathbf{p})\right)_{i=0}^{M},-\mathbf{g}(\mathbf{p})\right)$
with $\mathbf{f}_{\boldsymbol{\mu},i}(\mathbf{p})$ defined in (\ref{dm:fx}).
Then, similar to Lemma \ref{lem:uvi}, $\bar{\mathcal{G}}$ is equivalent
to $\mathrm{VI}(\bar{\mathcal{S}},\bar{\mathbf{f}})$, which has the
following property.
\begin{lem}
\label{lem:mon}Given $\Psi\succeq\mathbf{0}$, $\mathrm{VI}(\bar{\mathcal{S}},\bar{\mathbf{f}})$
is a monotone VI.
\end{lem}
\begin{proof}
The proof follows similar steps as used in the proof of Proposition
\ref{pro:unq} and exploits the monotonicity definition in Appendix
\ref{subsec:vi}.
\end{proof}
The monotonicity enables us to exploit methods from VI theory to solve
$\mathrm{VI}(\bar{\mathcal{S}},\bar{\mathbf{f}})$. An efficient method
is the proximal point method \cite{FacchineiPang03}, which employs
the following iteration 
\[
\bar{\mathbf{p}}^{k+1}=(1-\eta_{k})\bar{\mathbf{p}}^{k}+\eta_{k}J_{c}(\bar{\mathbf{p}}^{k}),
\]
where $J_{c}(\bar{\mathbf{p}}^{k})$ is the solution to $\mathrm{VI}(\bar{\mathcal{S}},\bar{\mathbf{f}}_{c,\bar{\mathbf{p}}^{k}})$
with $\bar{\mathbf{f}}_{c,\bar{\mathbf{p}}^{k}}(\bar{\mathbf{p}})\triangleq\bar{\mathbf{f}}(\bar{\mathbf{p}})+c(\bar{\mathbf{p}}-\bar{\mathbf{p}}^{k})$.
It is not difficult to see that $\mathrm{VI}(\bar{\mathcal{S}},\bar{\mathbf{f}}_{c,\bar{\mathbf{p}}^{k}})$
is actually equivalent to NEP $\bar{\mathcal{G}}_{c,k}$ below 
\begin{equation}
\underset{\mathbf{p}_{i}\in\mathcal{S}_{i}^{\mathrm{pow}}}{\mathrm{maximize}}\;R_{i}(\mathbf{p}_{i},\mathbf{p}_{-i})-\boldsymbol{\mu}^{T}\mathbf{g}(\mathbf{p})-\frac{c}{2}\left\Vert \mathbf{p}_{i}-\mathbf{p}_{i}^{k}\right\Vert ^{2}\label{da:pp}
\end{equation}
for $i=0,\ldots,M$ and
\begin{equation}
\underset{\boldsymbol{\mu}\geq\mathbf{0}}{\mathrm{maximize}}\;\boldsymbol{\mu}^{T}\mathbf{g}(\mathbf{p})-\frac{c}{2}\left\Vert \boldsymbol{\mu}-\boldsymbol{\mu}^{k}\right\Vert ^{2}.\label{da:up}
\end{equation}
Hence, $J_{c}(\bar{\mathbf{p}}^{k})$ is given by the NE of $\bar{\mathcal{G}}_{c,k}$. 

Parameter $c$ in (\ref{da:pp}) is chosen large enough such that
$\mathrm{VI}(\bar{\mathcal{S}},\bar{\mathbf{f}}_{c,\bar{\mathbf{p}}^{k}})$
is strongly monotone or equivalently the objectives in (\ref{da:pp})
and (\ref{da:up}) are all strongly concave. In this case, $\bar{\mathcal{G}}_{c,k}$
has a unique NE that can be computed via a best response algorithm
with guaranteed convergence. To this end, we derive the closed-form
solutions of (\ref{da:pp}) and (\ref{da:up}) as 
\begin{multline}
p_{i}^{\star}(n)=\bar{V}_{i,n}(\mathbf{p}_{-i},\boldsymbol{\mu})\triangleq\left[\frac{B_{i,n}(\mathbf{p}_{-i},\lambda_{i})}{A_{i,n}}+\right.\\
\left.\frac{\sqrt{B_{i,n}^{2}(\mathbf{p}_{-i},\lambda_{i})+2A_{i,n}C_{i,n}(\mathbf{p}_{-i},\lambda_{i})}}{A_{i,n}}\right]_{0}^{p_{i,n}^{\mathrm{peak}}}\label{da:opn}
\end{multline}
and $\mu_{n}^{\star}=\bar{U}_{n}(\mathbf{p})\triangleq\left[\frac{g_{n}\left(\mathbf{p}\right)}{c}+\mu_{n}^{k}\right]_{+}$
with $A_{i,n}\triangleq2c\,h_{ii}(n)$, $B_{i,n}(\mathbf{p}_{-i},\lambda_{i})\triangleq h_{ii}(n)\phi_{i,n}(\lambda_{i})-cI_{i,n}(\mathbf{p}_{-i})$,
$C_{i,n}(\mathbf{p}_{-i},\lambda_{i})\triangleq I_{i,n}(\mathbf{p}_{-i})\phi_{i,n}(\lambda_{i})+h_{ii}(n)$,
and
\[
\phi_{i,n}(\lambda_{i})\triangleq\begin{cases}
cp_{0}^{k}(n)+\mu_{n}\tilde{h}_{00}(n)-\lambda_{0}, & i=0\\
cp_{i}^{k}(n)-\mu_{n}h_{i0}(n)-\lambda_{i}, & i\neq0
\end{cases}
\]
where $\lambda_{i}$ is the minimum value such that $\sum_{n=1}^{N}p_{i}^{\star}(n)\leq p_{i}^{\mathrm{sum}}$
for $\forall i$. To find the optimal $\lambda_{i}$, we provide the
following result.
\begin{lem}
\label{lem:plm}$p_{i}^{\star}(n)$ is monotonically nonincreasing
in $\lambda_{i}$.
\end{lem}
\begin{proof}
See Appendix \ref{subsec:plm}.
\end{proof}
Therefore, one can exploit the bisection method to determine the optimal
$\lambda_{i}$. Define the mappings $\bar{\mathbf{V}}_{i}\left(\mathbf{p}_{-i},\boldsymbol{\mu}\right)\triangleq\left(\bar{V}_{i,n}(\mathbf{p}_{-i},\boldsymbol{\mu})\right)_{n=1}^{N}$
and $\bar{\mathbf{U}}(\mathbf{p})\triangleq\left(\bar{U}_{n}(\mathbf{p})\right)_{n=1}^{N}$.
Then, the NE of $\bar{\mathcal{G}}_{c,k}$ can be distributively computed
via Algorithm 1 by replacing $\mathbf{p}_{i}^{t+1}=\mathbf{V}_{i}\left(\mathbf{p}_{-i}^{t},\boldsymbol{\mu}\right)$
with $\boldsymbol{\mu}^{t+1}=\bar{\mathbf{U}}(\mathbf{p}^{t})$ and
$\mathbf{p}_{i}^{t+1}=\bar{\mathbf{V}}_{i}\left(\mathbf{p}_{-i}^{t},\boldsymbol{\mu}^{t}\right)$
for $i=0,...,M$. With the above VI and NEP interpretations, we formally
state the proximal point method applied to solving GNEP $\mathcal{G}$
in Algorithm 3. The convergence of Algorithm 3 is investigated in
Theorem \ref{thm:dpv}.

\begin{algorithm}[tbph]
\caption{\textbf{: Distributed Proximal Algorithm for }$\mathcal{G}$}

1: Set the initial point $(\mathbf{p}^{0},\boldsymbol{\mu}^{0})$,
precision $\epsilon$, and $k=0$;

2: Compute the NE $(\mathbf{p}^{\star k},\boldsymbol{\mu}^{\star k})$
of $\bar{\mathcal{G}}_{c,k}$ via Algorithm 1;

3: Update the power and price as $\mathbf{p}^{k+1}=(1-\eta_{k})\mathbf{p}^{k}+\eta_{k}\mathbf{p}^{\star k}$,
$\boldsymbol{\mu}^{k+1}=(1-\eta_{k})\boldsymbol{\mu}^{k}+\eta_{k}\boldsymbol{\mu}^{\star k}$;

4: $k=k+1$;

5: If $\left\Vert \bar{\mathbf{p}}^{k}-\bar{\mathbf{p}}^{k-1}\right\Vert \leq\epsilon$
stop, otherwise go to step 2.
\end{algorithm}

\begin{thm}
\label{thm:dpv}Given $\mathbf{\Psi}\succeq\mathbf{0}$ and $0<\eta_{k}<2$
for $\forall k$, the sequence $\{(\mathbf{p}^{k},\boldsymbol{\mu}^{k})\}_{k=0}^{\infty}$
generated by Algorithm 3 converges to a solution $(\mathbf{p}^{\star},\boldsymbol{\mu}^{\star})$
of $\mathrm{VI}(\bar{\mathcal{S}},\bar{\mathbf{f}})$\textup{ and
$\mathbf{p}^{\star}$ is a GNE (VE) of }$\mathcal{G}$.
\end{thm}
\begin{proof}
Theorem \ref{thm:dpv} follows from Lemma \ref{lem:mon} and \cite[Th. 12.3.9]{FacchineiPang03}.
\end{proof}
Similar to Algorithm 2, Algorithm 3 is also a distributed algorithm,
since the update of the proximal point can be conducted locally at
a BS or MUE. An advantage of Algorithm 3 is that price and power are
updated simultaneously, which may accelerate the convergence of the
price. Another advantage is that the condition $\mathbf{\Psi}\succeq\mathbf{0}$
in Theorem \ref{thm:dpv} is a bit weaker than requiring that $\mathbf{\Psi}$
is a P-matrix in Theorem \ref{thm:pcv}. On the other hand, Algorithm
3 still has two loops, the inner one for computing the NE of $\bar{\mathcal{G}}_{c,k}$
and the outer one for updating the proximal point. Because of the
simultaneous updating, the price in the inner loop is updated more
frequently than in Algorithm 2, meaning that the MUEs have to broadcast
the price more frequently. This implies a tradeoff between convergence
speed and signaling overhead.

\begin{figure*}[b]
\centering \hrulefill \setcounter{MYtempeqncnt}{\value{equation}}
\setcounter{equation}{23} 
\begin{equation}
\left[\mathbf{\Upsilon}\right]_{ij}\triangleq\begin{cases}
\max_{n}\sum_{l\neq j}h_{jl}^{2}(n)h_{ll}(n)p_{l,n}^{\mathrm{max}}\frac{\left(I_{l,n}(\mathbf{p}_{-l}^{\mathrm{max}})+I_{l,n}(\mathbf{p}^{\mathrm{max}})\right)^{2}}{\sigma_{l}^{4}(n)}, & i=j\\
\max_{n}h_{ij}(n)\frac{h_{jj}(n)}{\sigma_{j}^{2}(n)}+\max_{n}\sum_{l\neq i,j}h_{il}(n)h_{jl}(n)\frac{\left(I_{l,n}(\mathbf{p}_{-l}^{\mathrm{max}})+I_{l,n}(\mathbf{p}^{\mathrm{max}})\right)^{2}}{\sigma_{l}^{4}(n)}, & i\neq j
\end{cases}\label{go:mx}
\end{equation}
\setcounter{equation}{\value{MYtempeqncnt}} 
\end{figure*}

\section{Network Utility Maximization via GNEP\label{sec:gog}}

In this section, we consider NUM problem $\mathcal{P}$ in (\ref{gf:sm}),
aiming to maximize the sum rate of all BSs under the global QoS constraints.
As pointed out in Section \ref{sec:gf}, $\mathcal{P}$ is an NP-hard
problem, i.e., finding its globally optimal solution requires prohibitive
computational complexity even if a centralized approach is used. Thus,
low-cost suboptimal solutions are preferable in practice. Our goal
is to develop efficient distributed methods to find a stationary solution
of $\mathcal{P}$ by utilizing the above introduced GNEP methods.

For this purpose, we establish a bridge between NUM problem $\mathcal{P}$
and a GNEP. Specifically, consider the following penalized GNEP:
\begin{equation}
\mathcal{G}_{\mathbf{p}^{u}}:\;\begin{array}{ll}
\underset{\mathbf{p}_{i}\in\mathcal{S}_{i}^{\mathrm{pow}}}{\mathrm{maximize}} & R_{i}(\mathbf{p}_{i},\mathbf{p}_{-i})-(\mathbf{p}_{i}-\mathbf{p}_{i}^{u})^{T}\mathbf{b}_{i}(\mathbf{p}^{u})\\
\mathrm{subject\,to} & \mathbf{g}(\mathbf{p})\leq\mathbf{0}
\end{array}\label{go:gp}
\end{equation}
for $i=0,\ldots,M$, where 
\begin{multline*}
\mathbf{b}_{i}(\mathbf{p})\triangleq-\sum_{j\neq i}\nabla_{\mathbf{p}_{i}}R_{j}(\mathbf{p})=-\sum_{j\neq i}\left(\frac{\partial R_{j,n}}{\partial p_{i}(n)}\right)_{n=1}^{N}=\\
\left(\sum_{j\neq i}\frac{h_{ij}(n)h_{jj}(n)p_{j}(n)}{I_{j,n}(\mathbf{p}_{-j})I_{j,n}(\mathbf{p})}\right)_{n=1}^{N}=\left(\sum_{j\neq i}\omega_{ij}(n)\right)_{n=1}^{N}
\end{multline*}
with $\omega_{ij}(n)\triangleq\frac{h_{ij}(n)h_{jj}(n)p_{j}(n)}{I_{j,n}(\mathbf{p}_{-j})I_{j,n}(\mathbf{p})}$
, $I_{j,n}(\mathbf{p}_{-j})\triangleq\sigma_{j}(n)+\sum_{l\neq j}h_{lj}(n)p_{l}(n)$,
and $I_{j,n}(\mathbf{p})\triangleq I_{j,n}(\mathbf{p}_{-j})+h_{jj}(n)p_{j}(n)$.
According to Proposition \ref{pro:unq}, $\mathcal{G}_{\mathbf{p}^{u}}$
always admits a VE, which is unique if $\mathbf{\Psi}$ is a P-matrix
or equivalently $\rho(\mathbf{\Phi})<1$.\footnote{Note that the term $(\mathbf{p}_{i}-\mathbf{p}_{i}^{u})^{T}\mathbf{b}_{i}(\mathbf{p}^{u})$
does not change the uniformly P property of the corresponding VI of
$\mathcal{G}_{\mathbf{p}^{u}}$.} We denote the VE of $\mathcal{G}_{\mathbf{p}^{u}}$ by $\mathrm{VE}(\mathbf{p}^{u})$,
which is the solution of $\mathrm{VI}(\mathcal{S},\mathbf{f}_{\mathbf{p}^{u}})$,
where $\mathbf{f}_{\mathbf{p}^{u}}(\mathbf{p})\triangleq\mathbf{f}(\mathbf{p})+\mathbf{b}(\mathbf{p}^{u})$
and $\mathbf{b}(\mathbf{p}^{u})\triangleq\left(\mathbf{b}_{i}(\mathbf{p}^{u})\right)_{i=0}^{M}$.
Then, $\mathcal{G}_{\mathbf{p}^{u}}$ is related to $\mathcal{P}$
in the following Proposition.
\begin{prop}
\label{pro:stp}A point $\mathbf{p}^{\star}$ is a stationary point
of $\mathcal{P}$ if and only if it is a fixed point of $\mathrm{VE}(\mathbf{p}^{u})$,
i.e., $\mathbf{p}^{\star}=\mathrm{VE}(\mathbf{p}^{\star})$.
\end{prop}
\begin{proof}
See Appendix \ref{subsec:stp}.
\end{proof}
Proposition \ref{pro:stp} suggests that we can achieve a stationary
point of $\mathcal{P}$ by using the fixed point iteration $\mathbf{p}^{u+1}=\mathrm{VE}(\mathbf{p}^{u})$,
where in each iteration a GNEP $\mathcal{G}_{\mathbf{p}^{u}}$ is
to be solved. Then, we can exploit the distributed pricing algorithm
(Algorithm 2) or the distributed proximal algorithm (Algorithm 3)
to obtain the VE of $\mathcal{G}_{\mathbf{p}^{u}}$. This procedure
is formally stated in Algorithm 4. 

\begin{algorithm}[tbph]
\caption{\textbf{: Distributed GNEP Algorithm for }$\mathcal{P}$}

1: Set the initial point $\mathbf{p}^{0}$, precision $\epsilon$,
and $u=0$;

2: Compute the $\mathrm{VE}(\mathbf{p}^{u})$ of $\mathcal{G}_{\mathbf{p}^{u}}$
via Algorithm 2 or 3;

3: Update the power as $\mathbf{p}^{u+1}=\mathrm{VE}(\mathbf{p}^{u})$;

4: $u=u+1$;

5: If $\left\Vert \mathbf{p}^{u}-\mathbf{p}^{u-1}\right\Vert \leq\epsilon$
stop, otherwise go to step 2.
\end{algorithm}

The iteration $\mathbf{p}^{u+1}=\mathrm{VE}(\mathbf{p}^{u})$ can
be performed locally at each BS and does not require any additional
signaling. On the other hand, each BS $i$ needs to know $\mathbf{b}_{i}(\mathbf{p}^{u})$
or equivalently $\omega_{ij}(n)=\frac{h_{ij}(n)h_{jj}(n)p_{j}(n)}{I_{j,n}(\mathbf{p}_{-j})I_{j,n}(\mathbf{p})}$
for $n=1,\ldots,N$ and $j\neq i$. $(\omega_{ij}(n))_{n=1}^{N}$
can be obtained by BS $j$ via feedback from its users. Then, BSs
$j\neq i$ have to forward $(\omega_{ij}(n))_{n=1}^{N}$ to BS $i$
via wireline or wireless backhaul links. Therefore, the cost of obtaining
a stationary point of $\mathcal{P}$, besides the additional computational
complexity, is an information exchange between BSs. This implies a
fundamental tradeoff between the network utility and signaling overhead.

To investigate the convergence of Algorithm 4, we introduce matrix
$\mathbf{\Upsilon}\in\mathbb{R}^{(M+1)\times(M+1)}$ shown in (\ref{go:mx})\addtocounter{equation}{1}
at the bottom of this page, where $p_{l,n}^{\mathrm{max}}\triangleq\min\{p_{l}^{\mathrm{sum}},p_{l,n}^{\mathrm{peak}}\}$,
$\mathbf{p}^{\mathrm{max}}\triangleq(p_{m,n}^{\mathrm{max}})_{n,m=0}^{N,M}$,
and $\mathbf{p}_{-l}^{\mathrm{max}}\triangleq(p_{m,n}^{\mathrm{max}})_{n=0,m\neq l}^{N,M}$.
Then, Algorithm 4 converges to a stationary solution of $\mathcal{P}$
under the following condition.
\begin{thm}
\label{thm:gcv}The sequence $\{\mathbf{p}^{u}\}_{u=0}^{\infty}$
generated by Algorithm 4 converges to a stationary point of $\mathcal{P}$
if $\rho(\mathbf{\Psi}^{-1}\mathbf{\Upsilon})<1$ and $\mathbf{\Psi}\succ\mathbf{0}$.
\end{thm}
\begin{proof}
See Appendix \ref{subsec:gcv}.
\end{proof}
The condition $\mathbf{\Psi}\succ\mathbf{0}$ (where $\mathbf{\Psi}$
is defined in (\ref{vr:ms})) in Theorem \ref{thm:gcv} is needed
to guarantee the convergence of the embedded Algorithm 2 or 3 (as
well as the invertibility of $\mathbf{\Psi}$). Meanwhile, Algorithm
4 needs one more condition, i.e., $\rho(\mathbf{\Psi}^{-1}\mathbf{\Upsilon})<1$,
for convergence. Considering the definition of $\mathbf{\Upsilon}$
in (\ref{go:mx}), this condition is more likely satisfied if the
cross interference between the BSs is small. Yet, the convergence
condition of Algorithm 4 is much stronger than those of Algorithms
2 and 3. One may wonder if we can relax the condition in Theorem \ref{thm:gcv}
to make the proposed GNEP algorithm more practical. 

The answer is positive. To this end, let us consider the following
GNEP: 
\begin{equation}
\begin{array}{ll}
\underset{\mathbf{p}_{i}\in\mathcal{S}_{i}^{\mathrm{pow}}}{\mathrm{maximize}} & R_{i}(\mathbf{p}_{i},\mathbf{p}_{-i})-(\mathbf{p}_{i}-\mathbf{p}_{i}^{u})^{T}\mathbf{b}_{i}(\mathbf{p}^{u})-\frac{\tau}{2}\left\Vert \mathbf{p}_{i}-\mathbf{q}_{i}^{v}\right\Vert ^{2}\\
\mathrm{subject\,to} & \mathbf{g}(\mathbf{p})\leq\mathbf{0}.
\end{array}\label{go:pp}
\end{equation}
We denote the GNEP in (\ref{go:pp}) by $\mathcal{G}_{\mathbf{p}^{u},\mathbf{q}^{v}}$,
which is obtained by adding the proximal term $-\frac{\tau}{2}\left\Vert \mathbf{p}_{i}-\mathbf{q}_{i}^{v}\right\Vert ^{2}$
to each objective in $\mathcal{G}_{\mathbf{p}^{u}}$. Denote the variational
equilibrium of $\mathcal{G}_{\mathbf{p}^{u},\mathbf{q}^{v}}$ by $\mathrm{VE}(\mathbf{p}^{u},\mathbf{q}^{v})$.
Then, we propose Algorithm 5 for finding a stationary solution of
$\mathcal{P}$ with guaranteed convergence.

\begin{algorithm}[tbph]
\caption{\textbf{: Distributed Proximal GNEP Algorithm for }$\mathcal{P}$}

1: Set the initial points $\mathbf{p}^{0},\mathbf{q}^{0}$, precision
$\epsilon$, and $u,v=0$;

2: Compute the $\mathrm{VE}(\mathbf{p}^{u},\mathbf{q}^{v})$ of $\mathcal{G}_{\mathbf{p}^{u},\mathbf{q}^{v}}$
via Algorithm 2 or 3;

3: Update $\mathbf{p}$ as $\mathbf{p}^{u+1}=\mathrm{VE}(\mathbf{p}^{u},\mathbf{q}^{v})$;

4: $u=u+1$;

5: If$\left\Vert \mathbf{p}^{u}-\mathbf{p}^{u-1}\right\Vert \leq\epsilon$
go to step 7, otherwise go to step 3;

6: Update $\mathbf{q}$ as $\mathbf{q}^{v+1}=(1-\kappa_{v})\mathbf{q}^{v}+\kappa_{v}\mathbf{p}^{u}$;

7: $v=v+1$;

8: If$\left\Vert \mathbf{q}^{v}-\mathbf{q}^{v-1}\right\Vert \leq\epsilon$
stop, otherwise $u=0$, go to step 2;
\end{algorithm}

To find the VE of $\mathcal{G}_{\mathbf{p}^{u},\mathbf{q}^{v}}$,
Algorithm 5 invokes Algorithm 2 or 3, which in turn invokes Algorithm
1, i.e., the best response algorithm. Particularly, if Algorithm 2
is invoked, the best response of each BS is obtained by solving 
\begin{multline*}
\underset{\mathbf{p}_{i}\in\mathcal{S}_{i}^{\mathrm{pow}}}{\mathrm{maximize}}\;R_{i}(\mathbf{p}_{i},\mathbf{p}_{-i})-(\mathbf{p}_{i}-\mathbf{p}_{i}^{u})^{T}\mathbf{b}_{i}(\mathbf{p}^{u})-\frac{\tau}{2}\left\Vert \mathbf{p}_{i}-\mathbf{q}_{i}^{v}\right\Vert ^{2}\\
-\boldsymbol{\mu}^{T}\mathbf{g}(\mathbf{p})
\end{multline*}
 whose solution is given by (\ref{da:opn}) with $A_{i,n}=2\tau h_{ii}(n)$,
$B_{i,n}(\mathbf{p}_{-i},\lambda_{i})=h_{ii}(n)\phi_{i,n}(\mathbf{p}_{-i},\lambda_{i})-\tau I_{i,n}(\mathbf{p}_{-i})$,
$C_{i,n}(\mathbf{p}_{-i},\lambda_{i})=I_{i,n}(\mathbf{p}_{-i})\phi_{i,n}(\mathbf{p}_{-i},\lambda_{i})+h_{ii}(n)$,
and 
\[
\phi_{i,n}(\mathbf{p}_{-i},\lambda_{i})=\begin{cases}
\tau q_{0}^{v}(n)-b_{0,n}^{u}+\mu_{n}\tilde{h}_{00}(n)-\lambda_{0}, & i=0\\
\tau q_{i}^{v}(n)-b_{i,n}^{u}-\mu_{n}h_{i0}(n)-\lambda_{i}, & i\neq0.
\end{cases}
\]
If Algorithm 3 is invoked, the best response of each BS is obtained
by solving 
\begin{multline*}
\underset{\mathbf{p}_{i}\in\mathcal{S}_{i}^{\mathrm{pow}}}{\mathrm{maximize}}\;R_{i}(\mathbf{p}_{i},\mathbf{p}_{-i})-(\mathbf{p}_{i}-\mathbf{p}_{i}^{u})^{T}\mathbf{b}_{i}(\mathbf{p}^{u})-\frac{\tau}{2}\left\Vert \mathbf{p}_{i}-\mathbf{q}_{i}^{v}\right\Vert ^{2}\\
-\boldsymbol{\mu}^{T}\mathbf{g}(\mathbf{p})-\frac{c}{2}\left\Vert \mathbf{p}_{i}-\mathbf{p}_{i}^{k}\right\Vert ^{2}
\end{multline*}
whose solution is still given in the form of (\ref{da:opn}) with
$A_{i,n}=2(\tau+c)h_{ii}(n)$, $B_{i,n}(\mathbf{p}_{-i},\lambda_{i})=h_{ii}(n)\phi_{i,n}(\mathbf{p}_{-i},\lambda_{i})-(\tau+c)I_{i,n}(\mathbf{p}_{-i})$,
$C_{i,n}(\mathbf{p}_{-i},\lambda_{i})=I_{i,n}(\mathbf{p}_{-i})\phi_{i,n}(\mathbf{p}_{-i},\lambda_{i})+h_{ii}(n)$,
and 
\begin{multline*}
\phi_{i,n}(\mathbf{p}_{-i},\lambda_{i})=\\
\begin{cases}
\tau q_{0}^{v}(n)+cp_{0}^{k}(n)-b_{0,n}^{u}+\mu_{n}\tilde{h}_{00}(n)-\lambda_{0}, & i=0\\
\tau q_{i}^{v}(n)+cp_{i}^{k}(n)-b_{i,n}^{u}-\mu_{n}h_{i0}(n)-\lambda_{i}, & i\neq0.
\end{cases}
\end{multline*}
In both cases, $\lambda_{i}$ is chosen to be the minimum value such
that $\sum_{n=1}^{N}p_{i}(n)\leq p_{i}^{\mathrm{sum}}$, $\forall i$.
Similar to Lemma \ref{lem:plm}, we can also show that $p_{i}^{\star}(n)$
is monotonically nonincreasing in $\lambda_{i}$ so that it can be
efficiently found via the bisection method.

Now, we study the convergence of Algorithm 5, which is quite involved.
Thus, we first investigate the inner iteration $\mathbf{p}^{u+1}=\mathrm{VE}(\mathbf{p}^{u},\mathbf{q}^{v})$
and provide the following useful result.
\begin{prop}
\label{pro:cpu} Given $\mathbf{q}^{v}$ and $\tau\geq\max\left\{ \left|\lambda_{\mathrm{min}}\left(\mathbf{\Psi}-\mathbf{\Upsilon}\right)\right|,\tau_{\mathbf{\Psi}}\right\} $
with $\tau_{\mathbf{\Psi}}\triangleq\max\{\max_{i}(\sum_{j\neq i}\left[\left|\mathbf{\Psi}\right|\right]_{ij}-\left[\mathbf{\Psi}\right]_{ii}),\max_{j}(\sum_{i}\left[\left|\mathbf{\Psi}\right|\right]_{ij}-\left[\mathbf{\Psi}\right]_{jj})\}$,
$\mathbf{p}^{u+1}=\mathrm{VE}(\mathbf{p}^{u},\mathbf{q}^{v})$ converges
to a stationary point of the following problem: 
\begin{equation}
\mathcal{P}_{\mathbf{q}^{v}}:\;\begin{array}{ll}
\underset{\mathbf{p}}{\mathrm{maximize}} & \sum_{i=0}^{M}R_{i}(\mathbf{p})-\frac{\tau}{2}\left\Vert \mathbf{p}-\mathbf{q}^{v}\right\Vert ^{2}\\
\mathrm{subject\,to} & \mathbf{p}_{i}\in\mathcal{S}_{i}^{\mathrm{pow}},\quad i=0,\ldots,M\\
 & R_{0,n}(\mathbf{p})\geq\gamma_{n},\quad n=1,\ldots,N.
\end{array}\label{pv:pv}
\end{equation}
\end{prop}
\begin{proof}
See Appendix \ref{subsec:cpu}.
\end{proof}
Problem $\mathcal{P}_{\mathbf{q}^{v}}$ is in fact the proximal version
of NUM problem $\mathcal{P}$ at $\mathbf{q}^{v}$. Proposition \ref{pro:cpu}
states that, if $\tau$ is chosen large enough, more exactly $\tau\geq\max\left\{ \left|\lambda_{\mathrm{min}}\left(\mathbf{\Psi}-\mathbf{\Upsilon}\right)\right|,\tau_{\mathbf{\Psi}}\right\} $,\footnote{Note that $\mathbf{\Psi}-\mathbf{\Upsilon}$ is a symmetric matrix,
so its eigenvalues are real.} the inner iteration will converge to a stationary point of $\mathcal{P}_{\mathbf{q}^{v}}$.
Let $F(\mathbf{p})\triangleq-\sum_{i=0}^{M}R_{i}(\mathbf{p})$ and
$F_{\tau,v}(\mathbf{p})\triangleq F(\mathbf{p})+\frac{\tau}{2}\left\Vert \mathbf{p}-\mathbf{q}^{v}\right\Vert ^{2}$.
The objective function in $\mathcal{P}_{\mathbf{q}^{v}}$ has the
following favorable property.
\begin{lem}
\label{lem:scv}Given $\tau>\left|\lambda_{\mathrm{min}}\left(\mathbf{\Psi}-\mathbf{\Upsilon}\right)\right|$,
$F_{\tau,v}(\mathbf{p})$ is strongly convex on $\mathcal{S}$, i.e.,
$(\mathbf{p}^{1}-\mathbf{p}^{2})^{T}(\nabla F_{\tau,v}(\mathbf{p}^{1})-\nabla F_{\tau,v}(\mathbf{p}^{2}))\geq L_{\mathrm{sc}}\left\Vert \mathbf{p}^{1}-\mathbf{p}^{2}\right\Vert ^{2}$,
$\forall\mathbf{p}^{1},\mathbf{p}^{2}\in\mathcal{S}$, with $L_{\mathrm{sc}}\triangleq\tau+\lambda_{\mathrm{min}}\left(\mathbf{\Psi}-\mathbf{\Upsilon}\right)$.
\end{lem}
\begin{proof}
See Appendix \ref{subsec:scv}.
\end{proof}
According to Lemma \ref{lem:scv}, $\mathcal{P}_{\mathbf{q}^{v}}$
actually becomes a convex problem with a unique solution if $\tau$
is chosen large enough, i.e., $\tau\geq\left|\lambda_{\mathrm{min}}\left(\mathbf{\Psi}-\mathbf{\Upsilon}\right)\right|$.
Therefore, if the inner iteration converges to $\mathbf{p}_{\mathbf{q}^{v}}^{\star}$,
i.e., $\mathbf{p}_{\mathbf{q}^{v}}^{\star}=\mathrm{VE}(\mathbf{p}^{\star},\mathbf{q}^{v})$,
$\mathbf{p}_{\mathbf{q}^{v}}^{\star}$ is the optimal solution of
$\mathcal{P}_{\mathbf{q}^{v}}$. The outer iteration $\mathbf{q}^{v+1}=(1-\kappa_{v})\mathbf{q}^{v}+\kappa_{v}\mathbf{p}_{\mathbf{q}^{v}}^{\star}$
is then a fixed point iteration of the optimal solution of $\mathcal{P}_{\mathbf{q}^{v}}$.
To prove its convergence, we need the following intermediate result.
\begin{lem}
\label{lem:lip}$\nabla F(\mathbf{p})$ is Lipschitz continuous on
$\mathcal{S}$, i.e., $\left\Vert \nabla F(\mathbf{p}^{1})-\nabla F(\mathbf{p}^{2})\right\Vert \leq L_{\mathrm{lip}}\left\Vert \mathbf{p}^{1}-\mathbf{p}^{2}\right\Vert $,
where $L_{\mathrm{lip}}\triangleq\sigma_{\max}\left(\mathbf{\Xi}\right)+\sigma_{\max}\left(\mathbf{\Upsilon}\right)$
with $\mathbf{\Xi}\in\mathbb{R}^{(M+1)\times(M+1)}$ defined as 
\[
\left[\mathbf{\Xi}\right]_{ij}\triangleq\begin{cases}
\max_{n}h_{ii}^{2}(n)\sigma_{i}^{-2}(n) & i=j\\
\max_{n}h_{ii}(n)h_{ji}(n)\sigma_{i}^{-2}(n) & i\neq j.
\end{cases}
\]
\end{lem}
\begin{proof}
See Appendix \ref{subsec:lip}.
\end{proof}
Then, the convergence of Algorithm 5 is provided below.
\begin{thm}
\label{thm:pgv}The sequence $\{\mathbf{q}^{v+1}\}_{v=0}^{\infty}$
generated by Algorithm 5 converges to a stationary point of $\mathcal{P}$,
if the following conditions are satisfied: 1) $\tau\geq\max\left\{ \left|\lambda_{\mathrm{min}}\left(\mathbf{\Psi}-\mathbf{\Upsilon}\right)\right|,\tau_{\mathbf{\Psi}}\right\} $,
2) $\kappa_{v}\in(0,1]$, 3) $\kappa_{v}<\min\left\{ \frac{2(\tau+\lambda_{\mathrm{min}}\left(\mathbf{\Psi}-\mathbf{\Upsilon}\right))}{L_{\mathrm{lip}}},\frac{\tau+\lambda_{\mathrm{min}}\left(\mathbf{\Psi}-\mathbf{\Upsilon}\right)}{2\tau+\lambda_{\mathrm{min}}\left(\mathbf{\Psi}-\mathbf{\Upsilon}\right)}\right\} $,
4) $\sum_{v}\kappa_{v}=\infty$.
\end{thm}
\begin{proof}
See Appendix \ref{subsec:pgv}.
\end{proof}
Theorem \ref{thm:pgv} indicates that we can always use Algorithm
5 to obtain a stationary point of $\mathcal{P}$ by setting the parameters
$\tau$ and $\kappa_{v}$ properly. Specifically, $\tau$ shall be
chosen large enough, as explicitly quantified in condition 1) in Theorem
\ref{thm:pgv}. Given $\tau$, one can always find a step size $\kappa_{v}$
satisfying conditions 2) to 4). Since the fixed point iterations $\mathbf{p}^{u+1}=\mathrm{VE}(\mathbf{p}^{u},\mathbf{q}^{v})$
and $\mathbf{q}^{v+1}=(1-\kappa_{v})\mathbf{q}^{v}+\kappa_{v}\mathbf{p}^{u}$
can be performed locally at each BS, similar to Algorithm 4, Algorithm
5 also enjoys a decentralized structure. Yet, the BSs have to exchange
their locally obtained information $\omega_{ij}(n)$, as an inevitable
cost of utility maximization of the entire network.

\begin{figure}[tbph]
\centering \includegraphics[width=2.5in]{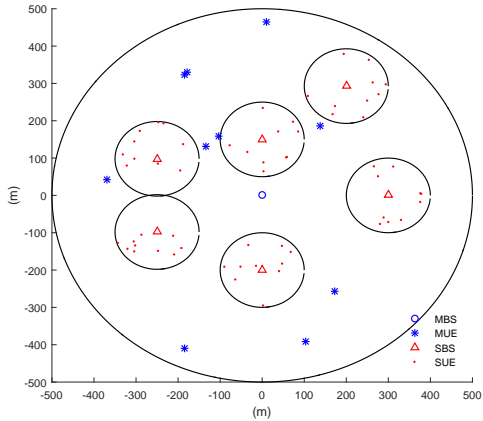} \caption{Topology of the small cell network.}

\label{fig:loc} 
\end{figure}

\section{Numerical Results\label{sec:nr}}

In this section, we evaluate the performance of the proposed GNEP
methods via numerical simulations. The MBS has $N=10$ channels shared
with $M=6$ SBSs, where each channel is allocated to one MUE and one
SUE in each macrocell and small cell, respectively, i.e., there are
$10$ MUEs and $60$ SUEs. The radii of the macrocell and the small
cells are 500m and 100m, respectively. The SBSs and MUEs are randomly
and uniformly located within the macrocell, and the SUEs are randomly
and uniformly located within each small cell, as shown in Fig. \ref{fig:loc}.
According to \cite{3GPP36814}, the path loss is given by $128.1+37.6\log_{10}d$
dB, where $d$ is the distance in kilometers. The small-scale fading
coefficients follow independent and identical zero-mean unit-variance
complex Gaussian distributions. We assume that only the total (sum)
power budgets are limited, which, from \cite{3GPP36814}, are set
to $p_{0}^{\mathrm{sum}}=46$dBm for the MBS and $p_{i}^{\mathrm{sum}}=33$dBm
for SBSs $i=1,...,M$. The noise power is $-114$dBm, corresponding
to a bandwidth of 10MHz and a noise power spectral density of $-174$dBm.

\begin{figure}[h]
\centering \includegraphics[width=2.9in]{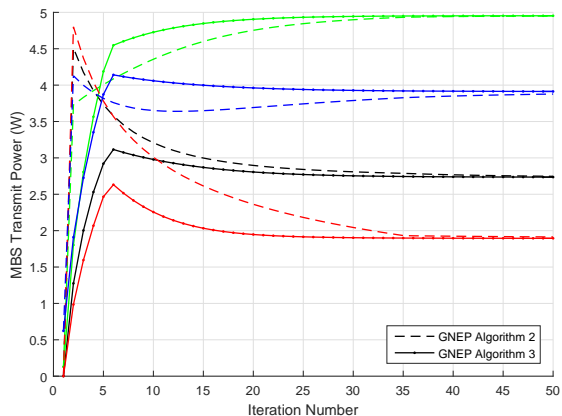} \caption{Convergence process of the transmit powers of the MBS for the two
GNEP algorithms.}
\label{fig:gcm} 
\end{figure}

\begin{figure}[h]
\centering \includegraphics[width=2.9in]{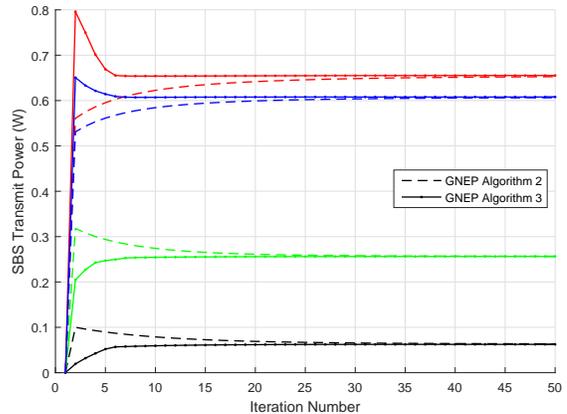} \caption{Convergence process of the transmit powers of one SBS for the two
GNEP algorithms.}
\label{fig:gcs} 
\end{figure}

For clarity, we refer to Algorithms 2 and 3 in Section \ref{sec:dcg}
as the GNEP methods since they aim to achieve a GNE of $\mathcal{G}$,
and to Algorithms 4 and 5 in Section \ref{sec:gog} as the NUM GNEP
methods since they aim to achieve a stationary solution of NUM problem
$\mathcal{P}$. The proposed methods are compared with two NEP-based
distributed methods, namely the NEP and QoS NEP methods. The NEP method
\cite{YuGinisCioffi02} is obtained by removing the global QoS constraints
from $\mathcal{G}$. In the QoS NEP method \cite{ScutariFacchinei14,WangScutariPalomar11},
the global QoS constraints are replaced by the individual QoS constraints
$h_{i0}(n)p_{i}(n)\leq\zeta_{i,n}$ for each SBS $i$ on channel $n$,
where the individual QoS threshold is set to $\zeta_{i,n}=(\tilde{h}_{00}(n)p_{0}(n)-\sigma_{0}(n))/M$
for $i=1,\ldots,M$ with $p_{0}(n)=p_{0}^{\mathrm{sum}}/N$. We also
compare the proposed methods with the interior point method \cite{BoydVan04},
i.e., a centralized optimization method, which can provide a stationary
solution to NUM problem $\mathcal{P}$.

In Figs. \ref{fig:gcm} and \ref{fig:gcs}, we display the convergence
process of the transmit powers for the two GNEP methods, i.e., Algorithms
2 and 3, versus the iteration number with QoS threshold $\gamma_{n}=2$
nats/s/Hz. Due to the large number of users, only the transmit powers
of four channels of the MBS and one SBS are shown. One can observe
that both algorithms converge rapidly to the same power allocation
profile, indicating that they achieve the same GNE. Compared to Algorithm
2, Algorithm 3 converges relatively faster, which is the benefit of
simultaneously updating transmit power and price. On the other hand,
Algorithm 2 enjoys the advantage of less signaling overhead, since
it requires the MBS to broadcast the price less frequently. This corresponds
to the commonly-observed tradeoff between convergence and information
exchange in distributed optimization.

\begin{figure}[h]
\centering \includegraphics[width=2.9in]{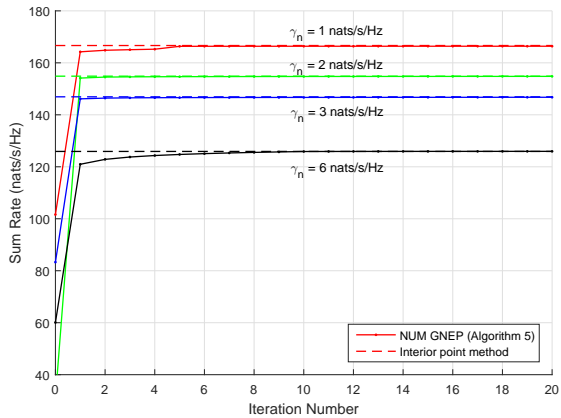} \caption{Convergence process of the sum rate of all BSs for the NUM GNEP and
interior point methods.}
\label{fig:gsv} 
\end{figure}

\begin{figure}[h]
\centering \includegraphics[width=2.9in]{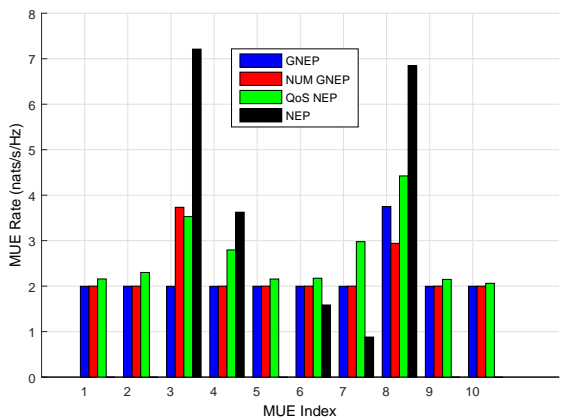} \caption{Rate of each MUE with QoS threshold $\gamma_{n}=2$ nats/s/Hz.}
\label{fig:muq} 
\end{figure}

Fig. \ref{fig:gsv} shows the convergence process of the sum rate
for the NUM GNEP method, i.e., Algorithm 5, versus the iteration number
for different QoS requirements. As indicated in Section \ref{sec:gog},
by properly choosing the algorithm parameters, Algorithm 5 is always
guaranteed to converge to a stationary solution of $\mathcal{P}$.
This is verified in Fig. \ref{fig:gsv}, where Algorithm 5 converges
to the same point as the interior point method. Note that the interior
point method is centralized and requires the collection of the channel
state information of the entire network at a central node, whereas
the NUM GNEP method can be implemented in a decentralized manner with
limited signaling overhead.

In Fig. \ref{fig:muq}, we show the rates of the MUEs generated by
different distributed methods with a QoS threshold of $\gamma_{n}=2$
nats/s/Hz for each MUE. The first observation is that the NEP method
may violate the QoS requirement and even result in zero rate for some
MUEs, i.e., it is not able to protect the macrocell communication.
The second observation is that, upon satifying the QoS requirement,
the GNEP and the NUM GNEP methods tend to meet exactly the QoS threshold,
whereas the QoS NEP method often leads to MUE rates higher than the
QoS threshold. From the system perspective, such redundancy in QoS
satisfaction may come at the expense of a degradation of the performance
of other utilities. e.g., the sum rate, of the BSs.

\begin{figure}[h]
\centering \includegraphics[width=2.9in]{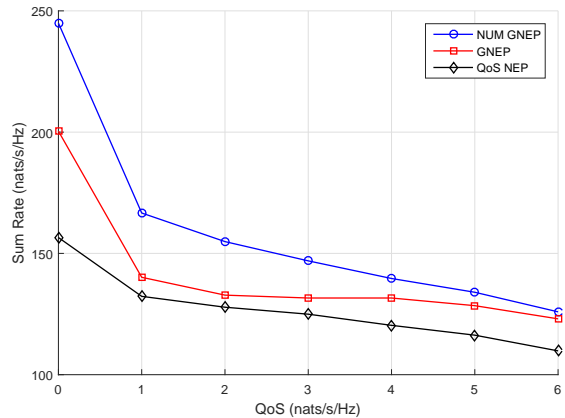} \caption{Sum rate versus QoS threshold.}
\label{fig:srq} 
\end{figure}

\begin{figure}[h]
\centering \includegraphics[width=2.9in]{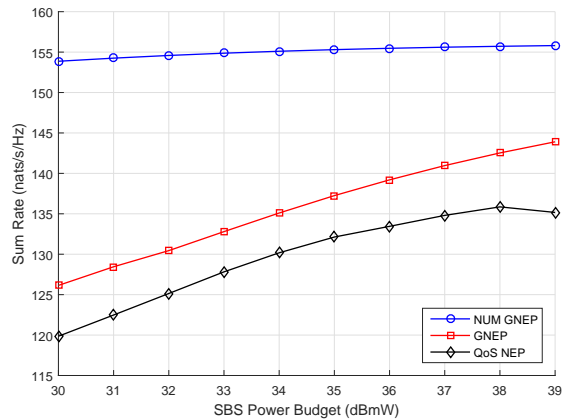} \caption{Sum rate versus SBS power budget.}
\label{fig:srp} 
\end{figure}

To further elaborate on this point, we plot the sum rate versus the
QoS threshold in Fig. \ref{fig:srq} and the sum rate versus the total
power budget of the SBSs with a QoS threshold of $\gamma_{n}=2$ nats/s/Hz
in Fig. \ref{fig:srp}. As can be observed from both figures, the
GNEP method achieves a higher sum rate than the QoS NEP method. This
is because the global QoS constraints allow to control the aggregate
interference at each MUE in a more flexible manner than the individual
QoS constraints and hence provide more degrees of freedom for improving
the network utility.

As expected, the NUM GNEP method, which aims to maximize the sum rate
of all BSs under the QoS constraints, achieves the highest sum rate.
The cost of the gain compared to the GNEP method is a higher signaling
overhead between the BSs as well as a higher complexity. Indeed, Figs.
\ref{fig:srq} and \ref{fig:srp} demonstrate the fundamental tradeoff
between signaling overhead and network performance, i.e., the more
signaling overhead can be afforded, the better the achievable network
performance. The proposed GNEP framework is able to flexibly adjust
this tradeoff by providing different GNEP-based distributed optimization
methods.

\section{Conclusions and Extensions\label{sec:cd}}

We have considered the interference management of a two-tier hierarchical
small cell network and developed a GNEP framework for distributed
optimization of the transmit strategies of the SBSs and the MBS. The
two different network design philosophies, i.e., achieving a GNE while
satisfying global QoS constraints and maximizing the network utility
under global QoS constraints, are unified under the GNEP framework.
We have developed various GNEP-based distributed algorithms, which
differ in the signaling overhead and complexity required to meet the
global QoS constraints at the GNE or to obtain a stationary solution
to the sum-rate maximization problem. The convergence properties of
the proposed algorithms were investigated.

The established GNEP framework can be extended in the following directions.
1) Multiple MBSs: Each MBS can logically view other MBSs as SBSs and
all BSs still play the same GNEP. 2) Global QoS constraints for SUEs:
the SBS, whose SUEs have global QoS requirements, can be logically
viewed as an MBS, which then falls into the case of multiple MBSs.
3) Other utilities: Each MBS or SBS can adopt utilities other than
the information rate as long as the utility is concave in its own
variables. The best response and convergence properties can be derived
using similar steps as in this paper.

The mathematical results derived in this paper guarantee that the
global QoS constraints are satisfied and a stationary solution is
obtained when the algorithms converge. Nevertheless, in practice,
one can terminate the algorithms before convergence (which is equivalent
to using a lower precision in the algorithms) and still obtain a close-to-optimal
performance. This and the aforementioned extensions make the desired
results applicable to a wide range of network scenarios including
dense and hyper-dense small cell networks.

\appendix{}

\subsection{A Brief Introduction to VI Theory\label{subsec:vi}}

We first introduce the basic concepts of VIs and generalized VIs (GVIs).
Specifically, let $\mathcal{X}\in\mathbb{R}^{n}$ and $\mathbf{F}:\mathcal{X}\rightarrow\mathbb{R}^{n}$
be a continuous function. Then, a VI problem, denoted by $\mathrm{VI}(\mathcal{X},\mathbf{F})$,
is to find a vector $\mathbf{x}^{\star}\in\mathcal{X}$ such that
$(\mathbf{x}-\mathbf{x}^{\star})^{T}\mathbf{F}(\mathbf{x}^{\star})\geq0$,
$\forall\mathbf{x}\in\mathcal{X}$. More generally, if $\mathbf{F}(\mathbf{x})$
is a point-to-set map $\mathcal{F}(\mathbf{x})\subseteq\mathbb{R}^{n}$
(also referred to as a set-valued function), then we have a GVI. Specifically,
let $\mathcal{X}\in\mathbb{R}^{n}$ and $\mathcal{F}:\mathcal{X}\rightarrow\mathbb{R}^{n}$
be a point-to-set map. Then, a GVI problem, denoted by $\mathrm{GVI}(\mathcal{X},\mathcal{F})$,
is to find a vector $\mathbf{x}^{\star}\in\mathcal{X}$ such that
there exists a vector $\mathbf{z}^{\star}\in\mathcal{F}(\mathbf{x}^{\star})$
and $(\mathbf{x}-\mathbf{x}^{\star})^{T}\mathbf{z}^{\star}\geq0$,
$\forall\mathbf{x}\in\mathcal{X}$.

In VI theory, a function $\mathbf{F}$ is monotone on $\mathcal{X}$
if for any distinct vectors $\mathbf{x},\mathbf{y}\in\mathcal{X}$,
$(\mathbf{x}-\mathbf{y})^{T}(\mathbf{F}(\mathbf{x})-\mathbf{F}(\mathbf{y}))\geq0,$
and strongly monotone if there exists a constant $c_{\mathrm{sm}}>0$
such that for any distinct vectors $\mathbf{x},\mathbf{y}\in\mathcal{X}$,
$\left(\mathbf{x}-\mathbf{y}\right)^{T}\left(\mathbf{F}(\mathbf{x})-\mathbf{F}(\mathbf{y})\right)\geq c_{\mathrm{sm}}\left\Vert \mathbf{x}-\mathbf{y}\right\Vert ^{2}$.
A function $\mathbf{F}\triangleq\left(\mathbf{F}_{k}\right)_{k=1}^{K}$
is uniformly P on $\mathcal{X}=\prod_{k=1}^{K}\mathcal{X}_{k}$ if
there exists a constant $c_{\mathrm{up}}>0$ such that for any distinct
vectors $\mathbf{x}=(\mathbf{x}_{k})_{k=1}^{K}$ and $\mathbf{y}=(\mathbf{y}_{k})_{k=1}^{K}\in\mathcal{X}$,
$\max_{k=1,\ldots,K}\left(\mathbf{x}_{k}-\mathbf{y}_{k}\right)^{T}\left(\mathbf{F}_{k}(\mathbf{x})-\mathbf{F}_{k}(\mathbf{y})\right)\geq c_{\mathrm{up}}\left\Vert \mathbf{x}-\mathbf{y}\right\Vert ^{2}$.
$\mathrm{VI}(\mathcal{X},\mathbf{F})$ is monotone, strongly monotone,
and uniformly P if $\mathbf{F}$ is monotone, strongly monotone, and
uniformly P on $\mathcal{X}$, respectively. Note that the uniformly
P property and strong monotonicity often imply a unique solution to
a VI (or GVI) \cite{FacchineiPang03}.

Related to the (strong) monotonicity and uniformly P property of VIs
are several special matrix definitions. A matrix $\mathbf{A}\in\mathbb{R}^{n\times n}$
is a Z-matrix if its off-diagonal entries are all non-positive and
a P-matrix if all its principle minors are positive. A Z-matrix that
is also a P-matrix is a K-matrix. These matrices have the following
properties.
\begin{lem}
\label{lem:pmx}(\cite{CottlePangStone92}) A matrix $\mathbf{A}\in\mathbb{R}^{n\times n}$
is a P-matrix if and only if $\mathbf{A}$ does not invert the sign
of any non-zero vector, i.e., if $x_{i}[\mathbf{A}\mathbf{x}]_{i}\leq0$
for $\forall i$, then $\mathbf{x}$ is an all-zero vector.
\end{lem}

\begin{lem}
\label{lem:kmx}(\cite{CottlePangStone92}) Let $\mathbf{A}\in\mathbb{R}^{n\times n}$
be a K-matrix and $\mathbf{B}$ a non-negative matrix. Then, $\rho(\mathbf{A}^{-1}\mathbf{B})<1$
if and only if $\mathbf{A}-\mathbf{B}$ is a K-matrix.
\end{lem}

\subsection{Proof of Proposition \ref{pro:unq}\label{subsec:unq}}

The existence of a solution of $\mathrm{VI}(\mathcal{S},\mathbf{f})$
follows directly from \cite[Corollary 2.2.5]{FacchineiPang03}. Next,
we show that $\mathbf{f}(\mathbf{p})$ is uniformly P on $\mathcal{S}$
if $\mathbf{\Psi}$ is a P-matrix, which leads to a unique solution.

Consider two distinct vectors $\mathbf{p}^{1}\triangleq(\mathbf{p}_{i}^{1})_{i=0}^{M},\mathbf{p}^{2}\triangleq(\mathbf{p}_{i}^{2})_{i=0}^{M}\in\mathcal{\mathcal{S}}$
and define $z_{i}(\theta)$ with $\theta\in[0,1]$ 
\begin{equation}
z_{i}(\theta)=\left(\mathbf{p}_{i}^{1}-\mathbf{p}_{i}^{2}\right)^{T}\mathbf{f}_{i}\left(\theta\mathbf{p}^{1}+(1-\theta)\mathbf{p}^{2}\right).\label{sm:fq}
\end{equation}
From the mean value theorem, the derivative of $z_{i}(\theta)$ with
respect to some $\theta\in[0,1]$ is given by 
\begin{equation}
z_{i}^{\prime}(\theta)=z_{i}(1)-z_{i}(0)=\left(\mathbf{p}_{i}^{1}-\mathbf{p}_{i}^{2}\right)^{T}\left(\mathbf{f}_{i}(\mathbf{p}^{1})-\mathbf{f}_{i}(\mathbf{p}^{2})\right).\label{sm:ft}
\end{equation}
On the other hand, $z_{i}^{\prime}(\theta)$ can be calculated as
\begin{align}
z_{i}^{\prime}(\theta) & =\left(\mathbf{p}_{i}^{1}-\mathbf{p}_{i}^{2}\right)^{T}\sum_{j=0}^{M}\nabla_{\mathbf{p}_{j}}\mathbf{f}_{i}(\mathbf{p}_{\theta})\left(\mathbf{p}_{j}^{1}-\mathbf{p}_{j}^{2}\right)\nonumber \\
 & =\mathbf{d}_{i}^{T}\sum_{j=0}^{M}\nabla_{\mathbf{p}_{j}}\mathbf{f}_{i}(\mathbf{p}_{\theta})\mathbf{d}_{j}\nonumber \\
 & \geq\mathbf{d}_{i}^{T}\nabla_{\mathbf{p}_{i}}\mathbf{f}_{i}(\mathbf{p}_{\theta})\mathbf{d}_{i}-\sum_{j\neq i}\left|\mathbf{d}_{i}^{T}\nabla_{\mathbf{p}_{j}}\mathbf{f}_{i}(\mathbf{p}_{\theta})\mathbf{d}_{j}\right|\label{sm:la}
\end{align}
where $\mathbf{d}_{j}\triangleq\mathbf{p}_{j}^{1}-\mathbf{p}_{j}^{2}$
and $\mathbf{p}_{\theta}\triangleq\theta\mathbf{p}^{1}+(1-\theta)\mathbf{p}^{2}\in\mathcal{\mathcal{S}}$
since $\mathcal{\mathcal{S}}$ is a convex set.

It is not difficult to verify that $\nabla_{\mathbf{p}_{j}}\mathbf{f}_{i}(\mathbf{p})=-\nabla_{\mathbf{p}_{i}\mathbf{p}_{j}}^{2}R_{i}(\mathbf{p})$
is an $N\times N$ diagonal matrix with diagonal elements 
\begin{multline}
h_{ii}(n)h_{ji}(n)\left(\sigma_{i}(n)+\sum_{l=0}^{M}h_{li}(n)p_{l}(n)\right)^{-2}\leq\\
h_{ii}(n)h_{ji}(n)\sigma_{i}^{-2}(n)\triangleq\beta_{ij}(n)\label{sm:bt}
\end{multline}
and $\nabla_{\mathbf{p}_{i}}\mathbf{f}_{i}(\mathbf{p})=-\nabla_{\mathbf{p}_{i}}^{2}R_{i}(\mathbf{p})$
is also an $N\times N$ diagonal matrix with diagonal elements 
\begin{multline}
h_{ii}^{2}(n)\left(\sigma_{i}(n)+\sum_{l=0}^{M}h_{li}(n)p_{l}(n)\right)^{-2}\geq\\
h_{ii}^{2}(n)\left(\sigma_{i}(n)+\sum_{l=0}^{M}h_{li}(n)p_{l,n}^{\mathrm{max}}\right)^{-2}\triangleq\alpha_{i}(n)\label{sm:af}
\end{multline}
where $p_{l,n}^{\mathrm{max}}\triangleq\min\left\{ p_{l}^{\mathrm{sum}},p_{l,n}^{\mathrm{peak}}\right\} $.
Let $\alpha_{i}^{\mathrm{min}}\triangleq\min_{n}\{\alpha_{i}(n)\}$
and $\beta_{ij}^{\mathrm{max}}\triangleq\max_{n}\{\beta_{ij}(n)\}$.
Then, we have $\mathbf{d}_{i}^{T}\nabla_{\mathbf{p}_{i}}\mathbf{f}_{i}(\mathbf{p}_{\theta})\mathbf{d}_{i}\geq\alpha_{i}^{\mathrm{min}}\left\Vert \mathbf{d}_{i}\right\Vert ^{2}$
and 
\begin{equation}
\left|\mathbf{d}_{i}^{T}\nabla_{\mathbf{p}_{j}}\mathbf{f}_{i}(\mathbf{p}_{\theta})\mathbf{d}_{j}\right|\leq\left\Vert \mathbf{d}_{i}\right\Vert \left\Vert \nabla_{\mathbf{p}_{j}}\mathbf{f}_{i}(\mathbf{p}_{\theta})\mathbf{d}_{j}\right\Vert \leq\beta_{ij}^{\mathrm{max}}\left\Vert \mathbf{d}_{i}\right\Vert \left\Vert \mathbf{d}_{j}\right\Vert .
\end{equation}

Consequently, we obtain 
\begin{align}
z_{i}^{\prime}(\theta) & \geq\alpha_{i}^{\mathrm{min}}\left\Vert \mathbf{d}_{i}\right\Vert ^{2}-\sum_{j\neq i}\beta_{ij}^{\mathrm{max}}\left\Vert \mathbf{d}_{i}\right\Vert \left\Vert \mathbf{d}_{j}\right\Vert \nonumber \\
 & =\alpha_{i}^{\mathrm{min}}s_{i}^{2}-\sum_{j\neq i}\beta_{ij}^{\mathrm{max}}s_{j}s_{i}=s_{i}\left[\mathbf{\Psi}\mathbf{s}\right]_{i}\label{sm:lb}
\end{align}
where $s_{j}\triangleq\left\Vert \mathbf{d}_{j}\right\Vert $, $\mathbf{s}\triangleq\left(s_{j}\right)_{j=0}^{M}$,
and $\mathbf{\Psi}$ is defined as in (\ref{vr:ms}). Combining (\ref{sm:ft})
and (\ref{sm:lb}), we have
\begin{equation}
\left(\mathbf{p}_{i}^{1}-\mathbf{p}_{i}^{2}\right)^{T}\left(\mathbf{f}_{i}(\mathbf{p}^{1})-\mathbf{f}_{i}(\mathbf{p}^{2})\right)\geq s_{i}\left[\mathbf{\Psi}\mathbf{s}\right]_{i}>0\label{sm:ff}
\end{equation}
where the last inequality follows from Lemma \ref{lem:pmx} in Appendix
\ref{subsec:vi}, since $\mathbf{\Psi}$ is a P-matrix. Therefore,
there always exists a constant $c_{\mathrm{up}}$ such that $\max_{i}\left(\mathbf{p}_{i}^{1}-\mathbf{p}_{i}^{2}\right)^{T}\left(\mathbf{f}_{i}(\mathbf{p}^{1})-\mathbf{f}_{i}(\mathbf{p}^{2})\right)\geq c_{\mathrm{up}}\left\Vert \mathbf{p}^{1}-\mathbf{p}^{2}\right\Vert ^{2}$,
which proves $\mathbf{f}(\mathbf{p})$ is a uniformly P function.

\subsection{Proof of Proposition \ref{pro:bes}\label{subsec:bes}}

From Lemma \ref{lem:uvi}, $\mathbf{p}^{\star}$ is a NE of $\mathcal{G}_{\boldsymbol{\mu}}$
if and only if 
\begin{equation}
\left(\mathbf{p}_{i}-\mathbf{p}_{i}^{\star}\right)^{T}\mathbf{f}_{\boldsymbol{\mu},i}(\mathbf{p}_{i}^{\star},\mathbf{p}_{-i}^{\star})\geq0,\quad\forall\mathbf{p}_{i}\in\mathcal{S}_{i}^{\mathrm{pow}}.\label{da:nc}
\end{equation}
From the first-order optimality condition, the best response $\mathbf{p}_{i}^{t+1}$
is a solution to (\ref{vr:gu}) with $\mathbf{p}_{-i}=\mathbf{p}_{-i}^{t}$
if and only if 
\begin{equation}
\left(\mathbf{p}_{i}-\mathbf{p}_{i}^{t+1}\right)^{T}\mathbf{f}_{\boldsymbol{\mu},i}(\mathbf{p}_{i},\mathbf{p}_{-i}^{t})\geq0,\quad\forall\mathbf{p}_{i}\in\mathcal{S}_{i}^{\mathrm{pow}}.\label{da:oc}
\end{equation}
Adding (\ref{da:nc}) with $\mathbf{p}_{i}=\mathbf{p}_{i}^{t+1}$
and (\ref{da:oc}) with $\mathbf{p}_{i}=\mathbf{p}_{i}^{\star}$,
we have 
\begin{align}
\left(\mathbf{p}_{i}^{t+1}-\mathbf{p}_{i}^{\star}\right)^{T}\left(\mathbf{f}_{\boldsymbol{\mu},i}(\mathbf{p}_{i}^{\star},\mathbf{p}_{-i}^{\star})-\mathbf{f}_{\boldsymbol{\mu},i}(\mathbf{p}_{i}^{t+1},\mathbf{p}_{-i}^{t})\right) & =\nonumber \\
\left(\mathbf{p}_{i}^{t+1}-\mathbf{p}_{i}^{\star}\right)^{T}\left(\mathbf{f}_{i}(\mathbf{p}_{i}^{\star},\mathbf{p}_{-i}^{\star})-\mathbf{f}_{i}(\mathbf{p}_{i}^{t+1},\mathbf{p}_{-i}^{t})\right) & \geq0\label{da:of}
\end{align}
where the equality follows from the definition of $\mathbf{f}_{\boldsymbol{\mu},i}$
in (\ref{dm:fx}).

Let $\mathbf{p}_{\theta}\triangleq\theta\left(\mathbf{p}_{i}^{\star},\mathbf{p}_{-i}^{\star}\right)+(1-\theta)\left(\mathbf{p}_{i}^{t+1},\mathbf{p}_{-i}^{t}\right)$.
It then follows from (\ref{sm:fq})-(\ref{sm:la}) and (\ref{sm:lb})
that for some $\theta\in[0,1]$ 
\begin{align}
0\leq & \left(\mathbf{p}_{i}^{t+1}-\mathbf{p}_{i}^{\star}\right)^{T}\left(\nabla_{\mathbf{p}_{i}}\mathbf{f}_{i}(\mathbf{p}_{\theta})(\mathbf{p}_{i}^{\star}-\mathbf{p}_{i}^{t+1})\right.\nonumber \\
 & +\sum_{j\neq i}\nabla_{\mathbf{p}_{j}}\mathbf{f}_{i}(\mathbf{p}_{\theta})\left.(\mathbf{p}_{j}^{\star}-\mathbf{p}_{j}^{t})\right)\nonumber \\
= & -(\mathbf{d}_{i}^{t+1})^{T}\nabla_{\mathbf{p}_{i}}\mathbf{f}_{i}(\mathbf{p}_{\theta})\mathbf{d}_{i}^{t+1}-\sum_{j\neq i}(\mathbf{d}_{i}^{t+1})^{T}\nabla_{\mathbf{p}_{j}}\mathbf{f}_{i}(\mathbf{p}_{\theta})\mathbf{d}_{j}^{t}\nonumber \\
\leq & -\alpha_{i}^{\mathrm{min}}\left\Vert \mathbf{d}_{i}^{t+1}\right\Vert ^{2}+\sum_{j\neq i}\beta_{ij}^{\mathrm{max}}\left\Vert \mathbf{d}_{i}^{t+1}\right\Vert \left\Vert \mathbf{d}_{j}^{t}\right\Vert .\label{da:ab}
\end{align}
where $\mathbf{d}_{i}^{t}\triangleq\mathbf{p}_{i}^{t}-\mathbf{p}_{i}^{\star}$,
and $\alpha_{i}^{\mathrm{min}}$ and $\beta_{ij}^{\mathrm{max}}$
are defined after (\ref{sm:af}). Introducing $s_{i}^{t}\triangleq\left\Vert \mathbf{d}_{i}^{t}\right\Vert $
and $\mathbf{s}^{t}\triangleq(s_{i}^{t})_{i=0}^{M}$, we obtain $s_{i}^{t+1}\leq\frac{1}{\alpha_{i}^{\mathrm{min}}}\sum_{j\neq i}\beta_{ij}^{\mathrm{max}}s_{j}^{t}$,
which, from the definition of $\mathbf{\Phi}$ in (\ref{vr:mf}),
implies $\mathbf{s}^{t+1}\leq\mathbf{\Phi}\mathbf{s}^{t}$. Therefore,
the error sequence $\left\{ \mathbf{s}^{t}\right\} $ converges to
zero if $\rho\left(\mathbf{\Phi}\right)<1$, which, from Lemma \ref{lem:sfe},
is equivalent to $\mathbf{\Psi}$ being a P-matrix.

\subsection{Proof of Lemma \ref{lem:plm}\label{subsec:plm}}

The optimal $p_{i}^{\star}(n)$ is obtained by taking the derivative
of the objective in (\ref{da:pp}), leading to 
\[
\frac{h_{00}(n)}{h_{00}(n)p_{0}(n)+I_{0,n}(\mathbf{p}_{-0})}-cp_{0}(n)+cp_{0}^{k}(n)+\mu_{n}\tilde{h}_{00}(n)=\lambda_{0}
\]
\[
\frac{h_{ii}(n)}{h_{ii}(n)p_{i}(n)+I_{i,n}(\mathbf{p}_{-i})}-cp_{i}(n)+cp_{i}^{k}(n)-\mu_{n}h_{i0}(n)=\lambda_{i}
\]
for $i=1,\ldots,M$. It is easily seen that $\lambda_{i}$ is monotonically
nonincreasing in $p_{i}(n)$. Since the projection does not change
the monotonicity, Lemma \ref{lem:plm} is thus proved.

\subsection{Proof of Proposition \ref{pro:stp}\label{subsec:stp}}

Since $\mathbf{f}_{\mathbf{p}^{u}}(\mathbf{p})=\left(-\nabla_{\mathbf{p}_{i}}R_{i}(\mathbf{p})+\mathbf{b}_{i}(\mathbf{p}^{u})\right)_{i=0}^{M}$,
if $\mathbf{p}^{\star}=\mathrm{VE}(\mathbf{p}^{\star})$, it follows
that 
\begin{align*}
(\mathbf{p}-\mathbf{p}^{\star})^{T}\mathbf{f}_{\mathbf{p}^{\star}}(\mathbf{p}^{\star}) & =\sum_{i=0}^{M}(\mathbf{p}_{i}-\mathbf{p}_{i}^{\star})^{T}\left(-\nabla_{\mathbf{p}_{i}}R_{i}(\mathbf{p}^{\star})+\mathbf{b}_{i}(\mathbf{p}^{\star})\right)\\
 & =\sum_{i=0}^{M}(\mathbf{p}_{i}-\mathbf{p}_{i}^{\star})^{T}\left(-\sum_{j=0}^{M}\nabla_{\mathbf{p}_{i}}R_{j}(\mathbf{p}^{\star})\right)\\
 & =-\sum_{j=0}^{M}(\mathbf{p}-\mathbf{p}^{\star})^{T}\nabla_{\mathbf{p}}R_{j}(\mathbf{p}^{\star})\geq0
\end{align*}
which is equivalent to the first-order optimality condition of $\mathcal{P}$,
implying $\mathbf{p}^{\star}$ is a stationary point of $\mathcal{P}$.
Conversely, if $\mathbf{p}^{\star}$ is a stationary point of $\mathcal{P}$,
we also obtain $\mathbf{p}^{\star}=\mathrm{VE}(\mathbf{p}^{\star})$.

\subsection{Proof of Theorem \ref{thm:gcv}\label{subsec:gcv}}

Consider two distinct points $\mathbf{p}^{1},\mathbf{p}^{2}\in\mathcal{S}$,
and let $\mathbf{z}^{1}=\mathrm{VE}(\mathbf{p}^{1})$ and $\mathbf{z}^{2}=\mathrm{VE}(\mathbf{p}^{2})$.
Then, we have for $\forall\mathbf{p}\in\mathcal{S}$, $\forall i$
\begin{align}
(\mathbf{p}_{i}-\mathbf{z}_{i}^{1})^{T}\left(\mathbf{f}_{i}(\mathbf{z}^{1})+\mathbf{b}_{i}(\mathbf{p}^{1})\right) & \geq0\label{gv:vea}\\
(\mathbf{p}_{i}-\mathbf{z}_{i}^{2})^{T}\left(\mathbf{f}_{i}(\mathbf{z}^{2})+\mathbf{b}_{i}(\mathbf{p}^{2})\right) & \geq0.\label{gv:veb}
\end{align}
By adding (\ref{gv:vea}) with $\mathbf{p}_{i}=\mathbf{z}_{i}^{2}$
to (\ref{gv:veb}) with $\mathbf{p}_{i}=\mathbf{z}_{i}^{1}$, we obtain
$(\mathbf{z}_{i}^{1}-\mathbf{z}_{i}^{2})^{T}(\mathbf{f}_{i}(\mathbf{z}^{2})-\mathbf{f}_{i}(\mathbf{z}^{1})+\mathbf{b}_{i}(\mathbf{p}^{2})-\mathbf{b}_{i}(\mathbf{p}^{1}))\geq0$
or equivalently $(\mathbf{z}_{i}^{1}-\mathbf{z}_{i}^{2})^{T}\left(\mathbf{b}_{i}(\mathbf{p}^{2})-\mathbf{b}_{i}(\mathbf{p}^{1})\right)\geq(\mathbf{z}_{i}^{1}-\mathbf{z}_{i}^{2})^{T}\left(\mathbf{f}_{i}(\mathbf{z}^{1})-\mathbf{f}_{i}(\mathbf{z}^{2})\right)$.
It follows from (\ref{sm:ff}) that 
\begin{equation}
(\mathbf{z}_{i}^{1}-\mathbf{z}_{i}^{2})^{T}\left(\mathbf{f}_{i}(\mathbf{z}^{1})-\mathbf{f}_{i}(\mathbf{z}^{2})\right)\geq s_{i}\left[\mathbf{\Psi}\mathbf{s}\right]_{i}\label{gv:zf}
\end{equation}
where $s_{i}\triangleq\left\Vert \mathbf{z}_{i}^{1}-\mathbf{z}_{i}^{2}\right\Vert $
and $\mathbf{s}\triangleq(s_{i})_{i=0}^{M}$.

Using the mean value theorem as in (\ref{sm:la}), we have for some
$\theta\in[0,1]$ and $\mathbf{p}_{\theta}\triangleq\theta\mathbf{p}^{2}+(1-\theta)\mathbf{p}^{1}\in\mathcal{\mathcal{S}}$
\begin{align}
 & (\mathbf{z}_{i}^{1}-\mathbf{z}_{i}^{2})^{T}\left(\mathbf{b}_{i}(\mathbf{p}^{2})-\mathbf{b}_{i}(\mathbf{p}^{1})\right)\nonumber \\
 & =\left(\mathbf{z}_{i}^{1}-\mathbf{z}_{i}^{2}\right)^{T}\sum_{j=0}^{M}\nabla_{\mathbf{p}_{j}}\mathbf{b}_{i}(\mathbf{p}_{\theta})(\mathbf{p}_{j}^{2}-\mathbf{p}_{j}^{1})\nonumber \\
 & \leq\sum_{j=0}^{M}\left\Vert \mathbf{z}_{i}^{1}-\mathbf{z}_{i}^{2}\right\Vert \left\Vert \nabla_{\mathbf{p}_{j}}\mathbf{b}_{i}(\mathbf{p}_{\theta})(\mathbf{p}_{j}^{2}-\mathbf{p}_{j}^{1})\right\Vert \nonumber \\
 & \leq\left\Vert \mathbf{z}_{i}^{1}-\mathbf{z}_{i}^{2}\right\Vert \sum_{j=0}^{M}\left\Vert \mathbf{p}_{j}^{2}-\mathbf{p}_{j}^{1}\right\Vert \sup_{\mathbf{p}_{\theta}}\left\Vert \nabla_{\mathbf{p}_{j}}\mathbf{b}_{i}(\mathbf{p}_{\theta})\right\Vert _{2}.\label{gv:zb}
\end{align}
Although complicated, it can be verified that for $\forall i,j$ 
\begin{equation}
\left[\mathbf{\Upsilon}\right]_{ij}\geq\sup_{\mathbf{p}_{\theta}}\left\Vert \nabla_{\mathbf{p}_{j}}\mathbf{b}_{i}(\mathbf{p}_{\theta})\right\Vert _{2}.\label{gv:fb}
\end{equation}
Let $\mathbf{e}\triangleq(e_{j})_{j=0}^{M}$ with $e_{j}\triangleq\left\Vert \mathbf{p}_{j}^{1}-\mathbf{p}_{j}^{2}\right\Vert $.
It follows from (\ref{gv:zf})-(\ref{gv:fb}) that $\sum_{j=0}^{M}\left[\mathbf{\Upsilon}\right]_{ij}e_{j}=\left[\mathbf{\Upsilon}\mathbf{e}\right]_{i}\geq\left[\mathbf{\Psi}\mathbf{s}\right]_{i}$,
which leads to $\mathbf{\Psi}\mathbf{s}\leq\mathbf{\Upsilon}\mathbf{e}$.
Since $\left\Vert \mathbf{s}\right\Vert =\left\Vert \mathbf{z}^{1}-\mathbf{z}^{2}\right\Vert $
and $\left\Vert \mathbf{e}\right\Vert =\left\Vert \mathbf{p}^{1}-\mathbf{p}^{2}\right\Vert $,
$\mathbf{z}=\mathrm{VE}(\mathbf{p})$ is a contraction mapping if
$\rho(\mathbf{\Psi}^{-1}\mathbf{\Upsilon})<1$. To guarantee the convergence
of Algorithms 2 and 3 as well as the invertibility of $\mathbf{\Psi}$,
we also require $\mathbf{\Psi}\succ\mathbf{0}$.

\subsection{Proof of Proposition \ref{pro:cpu}\label{subsec:cpu}}

Similar to Proposition \ref{pro:stp}, given $\mathbf{q}^{v}$, if
$\mathbf{p}^{u+1}=\mathrm{VE}(\mathbf{p}^{u},\mathbf{q}^{v})$ converges
to $\mathbf{p}_{\mathbf{q}^{v}}^{\star}$, it must be a stationary
point of $\mathcal{P}_{\mathbf{q}^{v}}$. Next, we prove the convergence
of $\mathbf{p}^{u+1}=\mathrm{VE}(\mathbf{p}^{u},\mathbf{q}^{v})$.

$\mathrm{VE}(\mathbf{p}^{u},\mathbf{q}^{v})$ is the solution of $\mathrm{VI}(\mathcal{S},\mathbf{f}_{\mathbf{p}^{u},\mathbf{q}^{v}})$
with $\mathbf{f}_{\mathbf{p}^{u}}(\mathbf{p})\triangleq\mathbf{f}(\mathbf{p})+\mathbf{b}(\mathbf{p}^{u})+\tau(\mathbf{p}-\mathbf{q}^{v})$.
Thus, given $\mathbf{z}^{1}=\mathrm{VE}(\mathbf{p}^{1},\mathbf{q}^{v})$
and $\mathbf{z}^{2}=\mathrm{VE}(\mathbf{p}^{2},\mathbf{q}^{v})$,
we have for $\forall\mathbf{p}\in\mathcal{S}$ and $\forall i$ 
\begin{align*}
(\mathbf{p}_{i}-\mathbf{z}_{i}^{1})^{T}\left(\mathbf{f}_{i}(\mathbf{z}^{1})+\mathbf{b}_{i}(\mathbf{p}^{1})+\tau(\mathbf{z}_{i}^{1}-\mathbf{q}_{i}^{v})\right) & \geq0\\
(\mathbf{p}_{i}-\mathbf{z}_{i}^{2})^{T}\left(\mathbf{f}_{i}(\mathbf{z}^{2})+\mathbf{b}_{i}(\mathbf{p}^{2})+\tau(\mathbf{z}_{i}^{2}-\mathbf{q}_{i}^{v})\right) & \geq0
\end{align*}
which leads to 
\begin{align*}
 & (\mathbf{z}_{i}^{1}-\mathbf{z}_{i}^{2})^{T}\left(\mathbf{b}_{i}(\mathbf{p}^{2})-\mathbf{b}_{i}(\mathbf{p}^{1})\right)\\
 & \geq(\mathbf{z}_{i}^{1}-\mathbf{z}_{i}^{2})^{T}\left(\mathbf{f}_{i}(\mathbf{z}^{1})-\mathbf{f}_{i}(\mathbf{z}^{2})\right)+\tau\left\Vert \mathbf{z}_{i}^{1}-\mathbf{z}_{i}^{2}\right\Vert ^{2}\\
 & \geq s_{i}\left[\mathbf{\Psi}\mathbf{s}\right]_{i}+\tau s_{i}^{2}
\end{align*}
where the last equality follows from (\ref{sm:ff}) with $s_{i}\triangleq\left\Vert \mathbf{z}_{i}^{1}-\mathbf{z}_{i}^{2}\right\Vert $
and $\mathbf{s}\triangleq(s_{i})_{i=0}^{M}$. Then, using (\ref{gv:zb})
and (\ref{gv:fb}), we have $\left[\mathbf{\Upsilon}\mathbf{e}\right]_{i}\geq\left[\mathbf{\Psi}\mathbf{s}\right]_{i}+\tau s_{i}$,
where $e_{i}\triangleq\left\Vert \mathbf{p}_{i}^{1}-\mathbf{p}_{i}^{2}\right\Vert $
and $\mathbf{e}\triangleq(e_{i})_{i=0}^{M}$, which leads to $(\tau\mathbf{I}+\mathbf{\Psi})\mathbf{s}\leq\mathbf{\Upsilon}\mathbf{e}$.
Thus, $\mathbf{z}=\mathrm{VE}(\mathbf{p},\mathbf{q}^{v})$ is a contraction
mapping if $\rho((\tau\mathbf{I}+\mathbf{\Psi})^{-1}\mathbf{\Upsilon})<1$
or equivalently (from Lemma \ref{lem:kmx}) $\tau\mathbf{I}+\mathbf{\Psi}-\mathbf{\Upsilon}$
is a P-matrix, which is implied by $\tau\mathbf{I}+\mathbf{\Psi}-\mathbf{\Upsilon}\succ\mathbf{0}$
and achieved by setting $\tau>\left|\lambda_{\mathrm{min}}\left(\mathbf{\Psi}-\mathbf{\Upsilon}\right)\right|$.
Meanwhile, for the convergence of Algorithms 2 and 3, we shall have
$\tau\mathbf{I}+\mathbf{\Psi}\succ\mathbf{0}$, which is implied by
the strict diagonal dominance, i.e., $\tau+\left[\mathbf{\Psi}\right]_{ii}>\sum_{j\neq i}\left[\left|\mathbf{\Psi}\right|\right]_{ij}$,
$\forall i$ and $\tau+\left[\mathbf{\Psi}\right]_{jj}>\sum_{i}\left[\left|\mathbf{\Psi}\right|\right]_{ij}$,
$\forall j$. Therefore, we have $\tau>\max\{\max_{i}(\sum_{j\neq i}\left[\left|\mathbf{\Psi}\right|\right]_{ij}-\left[\mathbf{\Psi}\right]_{ii}),\max_{j}(\sum_{i}\left[\left|\mathbf{\Psi}\right|\right]_{ij}-\left[\mathbf{\Psi}\right]_{jj})\}$.

\subsection{Proof of Lemma \ref{lem:scv}\label{subsec:scv}}

Since $\nabla_{\mathbf{p}_{i}}F_{\tau}(\mathbf{p})=\mathbf{f}_{i}(\mathbf{p})+\mathbf{b}_{i}(\mathbf{p})+\tau(\mathbf{p}_{i}-\mathbf{q}_{i}^{v})$,
we have
\begin{multline*}
(\mathbf{p}^{1}-\mathbf{p}^{2})^{T}(\nabla_{\mathbf{p}}F_{\tau}(\mathbf{p}^{1})-\nabla_{\mathbf{p}}F_{\tau}(\mathbf{p}^{2}))=\sum_{i=0}^{M}\tau\left\Vert \mathbf{p}_{i}^{1}-\mathbf{p}_{i}^{2}\right\Vert ^{2}\\
+\sum_{i=0}^{M}(\mathbf{p}_{i}^{1}-\mathbf{p}_{i}^{2})^{T}\left(\mathbf{f}_{i}(\mathbf{p}^{1})-\mathbf{f}_{i}(\mathbf{p}^{2})+\mathbf{b}_{i}(\mathbf{p}^{1})-\mathbf{b}_{i}(\mathbf{p}^{2})\right).
\end{multline*}
With $\mathbf{d}_{i}\triangleq\mathbf{p}_{i}^{1}-\mathbf{p}_{i}^{2}$
and $s_{i}\triangleq\left\Vert \mathbf{d}_{i}\right\Vert $ and from
(\ref{sm:fq})-(\ref{sm:la}), we have for some $\theta\in[0,1]$
and $\mathbf{p}_{\theta}\triangleq\theta\mathbf{p}^{1}+(1-\theta)\mathbf{p}^{2}\in\mathcal{\mathcal{S}}$
\begin{alignat*}{1}
 & (\mathbf{p}_{i}^{1}-\mathbf{p}_{i}^{2})^{T}(\mathbf{b}_{i}(\mathbf{p}^{1})-\mathbf{b}_{i}(\mathbf{p}^{2})\\
 & \geq\mathbf{d}_{i}^{T}\nabla_{\mathbf{p}_{i}}\mathbf{b}_{i}(\mathbf{p}_{\theta})\mathbf{d}_{i}-\sum_{j\neq i}\left|\mathbf{d}_{i}^{T}\nabla_{\mathbf{p}_{j}}\mathbf{b}_{i}(\mathbf{p}_{\theta})\mathbf{d}_{j}\right|\\
 & \geq s_{i}^{2}\inf_{\mathbf{p}_{\theta}}\lambda_{\mathrm{min}}\left(\nabla_{\mathbf{p}_{i}}\mathbf{b}_{i}(\mathbf{p}_{\theta})\right)-\sum_{j\neq i}s_{i}s_{j}\sup_{\mathbf{p}_{\theta}}\left\Vert \nabla_{\mathbf{p}_{j}}\mathbf{b}_{i}(\mathbf{p}_{\theta})\right\Vert _{2}.
\end{alignat*}
It can be verified that $\inf_{\mathbf{p}_{\theta}}\lambda_{\mathrm{min}}\left(\nabla_{\mathbf{p}_{i}}\mathbf{b}_{i}(\mathbf{p}_{\theta})\right)\geq-\left[\mathbf{\Upsilon}\right]_{ii}$
and $\sup_{\mathbf{p}_{\theta}}\left\Vert \nabla_{\mathbf{p}_{j}}\mathbf{b}_{i}(\mathbf{p}_{\theta})\right\Vert \leq\left[\mathbf{\Upsilon}\right]_{ij}$
for $j\neq i$. So, we have $(\mathbf{p}_{i}^{1}-\mathbf{p}_{i}^{2})^{T}\left(\mathbf{b}_{i}(\mathbf{p}^{1})-\mathbf{b}_{i}(\mathbf{p}^{2})\right)\geq-s_{i}\left[\mathbf{\Upsilon}\mathbf{s}\right]_{i}$.
Meanwhile, from (\ref{sm:ff}), we also have $(\mathbf{p}_{i}^{1}-\mathbf{p}_{i}^{2})^{T}\left(\mathbf{f}_{i}(\mathbf{p}^{1})-\mathbf{f}_{i}(\mathbf{p}^{2})\right)\geq s_{i}\left[\mathbf{\Psi}\mathbf{s}\right]_{i}$.
Consequently, we obtain 
\begin{align*}
 & (\mathbf{p}^{1}-\mathbf{p}^{2})^{T}\left(\nabla_{\mathbf{p}}F_{\tau}(\mathbf{p}^{1})-\nabla_{\mathbf{p}}F_{\tau}(\mathbf{p}^{2})\right)\\
 & \geq\sum_{i=0}^{M}\left(\tau s_{i}^{2}+s_{i}\left[\mathbf{\Psi}\mathbf{s}\right]_{i}-s_{i}\left[\mathbf{\Upsilon}\mathbf{s}\right]_{i}\right)\\
 & =\mathbf{s}^{T}\left(\tau\mathbf{I}+\mathbf{\Psi}-\mathbf{\Upsilon}\right)\mathbf{s}\geq\left\Vert \mathbf{s}\right\Vert ^{2}\lambda_{\mathrm{min}}\left(\tau\mathbf{I}+\mathbf{\Psi}-\mathbf{\Upsilon}\right).
\end{align*}
Therefore, $F_{\tau}(\mathbf{p})$ is strongly convex if $\tau\mathbf{I}+\mathbf{\Psi}-\mathbf{\Upsilon}\succ\mathbf{0}$,
which is implied by $\tau>\left|\lambda_{\mathrm{min}}\left(\mathbf{\Psi}-\mathbf{\Upsilon}\right)\right|$.
The strong convexity constant $L_{\mathrm{sc}}$ in Lemma \ref{lem:scv}
is then given by $\tau+\lambda_{\mathrm{min}}\left(\mathbf{\Psi}-\mathbf{\Upsilon}\right)$.

\subsection{Proof of Lemma \ref{lem:lip}\label{subsec:lip}}

According to the mean value theory, we have for some $\theta\in[0,1]$
and $\mathbf{p}_{\theta}=\theta\mathbf{p}^{1}-(1-\theta)\mathbf{p}^{2}\in\mathcal{S}$
$\left\Vert \nabla_{\mathbf{p}}F(\mathbf{p}^{1})-\nabla_{\mathbf{p}}F(\mathbf{p}^{2})\right\Vert \leq\left\Vert \nabla_{\mathbf{p}}^{2}F(\mathbf{p}_{\theta})\right\Vert _{2}\left\Vert \mathbf{p}^{1}-\mathbf{p}^{2}\right\Vert $.
Since $\nabla_{\mathbf{p}}^{2}F(\mathbf{p})=\left[\nabla_{\mathbf{p}_{i}\mathbf{p}_{j}}F(\mathbf{p})\right]_{i,j=0}^{M}=\left[\nabla_{\mathbf{p}_{j}}\mathbf{f}_{i}(\mathbf{p})\right]_{i,j=0}^{M}+\left[\nabla_{\mathbf{p}_{j}}\mathbf{b}_{i}(\mathbf{p})\right]_{i,j=0}^{M}$,
it follows that $\left\Vert \nabla_{\mathbf{p}}^{2}F(\mathbf{p})\right\Vert _{2}\leq\sigma_{\max}\left(\left[\nabla_{\mathbf{p}_{j}}\mathbf{f}_{i}(\mathbf{p})\right]_{i,j=0}^{M}\right)+\sigma_{\max}\left(\left[\nabla_{\mathbf{p}_{j}}\mathbf{b}_{i}(\mathbf{p})\right]_{i,j=0}^{M}\right)$.
Furthermore, we have $\sigma_{\max}\left(\left[\nabla_{\mathbf{p}_{j}}\mathbf{f}_{i}(\mathbf{p})\right]_{i,j=0}^{M}\right)\leq\sigma_{\max}\left(\left[\sigma_{\max}\left(\nabla_{\mathbf{p}_{j}}\mathbf{f}_{i}(\mathbf{p})\right)\right]_{i,j=0}^{M}\right)$.
It can be verified that $\sigma_{\max}\left(\nabla_{\mathbf{p}_{j}}\mathbf{f}_{i}(\mathbf{p})\right)\leq\left[\mathbf{\Xi}\right]_{ij}$,
implying $\sigma_{\max}\left(\left[\nabla_{\mathbf{p}_{j}}\mathbf{f}_{i}(\mathbf{p})\right]_{i,j=0}^{M}\right)\leq\sigma_{\max}\left(\mathbf{\Xi}\right)$.
Similarly, we can also obtain $\sigma_{\max}\left(\left[\nabla_{\mathbf{p}_{j}}\mathbf{b}_{i}(\mathbf{p})\right]_{i,j=0}^{M}\right)\leq\sigma_{\max}\left(\mathbf{\Upsilon}\right)$.
Therefore, we have $\left\Vert \nabla_{\mathbf{p}}^{2}F(\mathbf{p}_{\theta})\right\Vert _{2}\leq\sigma_{\max}\left(\mathbf{\Xi}\right)+\sigma_{\max}\left(\mathbf{\Upsilon}\right)$,
which is the Lipschitz constant.

\subsection{Proof of Theorem \ref{thm:pgv}\label{subsec:pgv}}

Let $\mathbf{p}_{\mathbf{q}^{v}}^{\star}=\mathrm{VE}(\mathbf{p}^{\star},\mathbf{q}^{v})$.
The strong convexity in Lemma \ref{lem:scv} leads to the following
useful results.
\begin{lem}
\label{lem:pqd} Given $\tau>\left|\lambda_{\mathrm{min}}\left(\mathbf{\Psi}-\mathbf{\Upsilon}\right)\right|$,
it follows that $\left\Vert \mathbf{p}_{\mathbf{q}^{1}}^{\star}-\mathbf{p}_{\mathbf{q}^{2}}^{\star}\right\Vert \leq\frac{\tau}{L_{\mathrm{sc}}}\left\Vert \mathbf{q}^{1}-\mathbf{q}^{2}\right\Vert $
and $(\mathbf{p}_{\mathbf{q}}^{\star}-\mathbf{q})^{T}\nabla_{\mathbf{p}}F(\mathbf{q})\leq-L_{\mathrm{sc}}\left\Vert \mathbf{p}_{\mathbf{q}}^{\star}-\mathbf{q}\right\Vert ^{2}$.
\end{lem}
Using the descent lemma \cite[Prop. A.24]{Bertsekas99} (with $\mathbf{x}=\mathbf{q}^{v}$
and $\mathbf{y}=\kappa_{v}(\mathbf{p}_{\mathbf{q}^{v}}^{\star}-\mathbf{q}^{v})$),
we have 
\begin{multline*}
F(\mathbf{q}^{v+1})\leq F(\mathbf{q}^{v})+\kappa_{v}(\mathbf{p}_{\mathbf{q}^{v}}^{\star}-\mathbf{q}^{v})^{T}\triangledown_{\mathbf{p}}F\left(\mathbf{q}^{v}\right)\\
+\frac{L_{\mathrm{lip}}\kappa_{v}^{2}}{2}\left\Vert \mathbf{p}_{\mathbf{q}^{v}}^{\star}-\mathbf{q}^{v}\right\Vert ^{2}.
\end{multline*}
It follows from Lemma \ref{lem:pqd} that 
\[
F(\mathbf{q}^{v+1})\leq F(\mathbf{q}^{v})+\frac{1}{2}\left(L_{\mathrm{lip}}\kappa_{v}^{2}-2\kappa_{v}L_{\mathrm{sc}}\right)\left\Vert \mathbf{p}_{\mathbf{q}^{v}}^{\star}-\mathbf{q}^{v}\right\Vert ^{2}.
\]
Given $L_{\mathrm{lip}}\kappa_{v}^{2}-2\kappa_{v}L_{\mathrm{sc}}<0$,
we have $F(\mathbf{q}^{v+1})\leq F(\mathbf{q}^{v})$, indicating that
$\{F(\mathbf{q}^{v})\}$ converges. This implies that 
\[
\sum_{v=1}^{\infty}\kappa_{v}\left\Vert \mathbf{p}_{\mathbf{q}^{v}}^{\star}-\mathbf{q}^{v}\right\Vert ^{2}<+\infty.
\]
Since $\sum_{v=1}^{\infty}\kappa_{v}=\infty$, we have $\liminf_{v\rightarrow\infty}\left\Vert \mathbf{p}_{\mathbf{q}^{v}}^{\star}-\mathbf{q}^{v}\right\Vert =0$.
Given $\kappa_{v}<\frac{1}{1+\frac{\tau}{L_{\mathrm{sc}}}}=\frac{\tau+\lambda_{\mathrm{min}}\left(\mathbf{\Psi}-\mathbf{\Upsilon}\right)}{2\tau+\lambda_{\mathrm{min}}\left(\mathbf{\Psi}-\mathbf{\Upsilon}\right)}$,
we can further prove that $\lim_{v\rightarrow\infty}\left\Vert \mathbf{p}_{\mathbf{q}^{v}}^{\star}-\mathbf{q}^{v}\right\Vert =0$,
for which we refer to \cite{Razaviyayn14a} for details. From Proposition
\ref{pro:cpu}, the fixed point is a stationary point of $\mathcal{P}$.

\bibliographystyle{IEEEtran}
\bibliography{IEEEabrv,jhwang}

\end{document}